\newcommand{\helios}{\textsc{Equinox}}

\newcommand{\niparagraph}[1]{\vspace{1.5mm}\noindent\textbf{#1}\enspace}

\documentclass[sigplan,10pt,nonacm]{acmart}
\renewcommand\footnotetextcopyrightpermission[1]{}
\AtBeginDocument{%
  }

\usepackage[table]{xcolor}
\usepackage{enumitem}
\usepackage{booktabs}
\usepackage{pgfplots}
\pgfplotsset{compat=1.18}
\usepackage{tikz}
\usepackage{microtype}
\usepackage{xspace}
\usepackage[ruled,linesnumbered]{algorithm2e}
\usepackage{balance}
\usepackage{gensymb}
\usepackage{subcaption}

\setcopyright{acmlicensed}
\copyrightyear{2018}
\acmYear{2018}
\acmDOI{XXXXXXX.XXXXXXX}
\acmConference[Conference acronym 'XX]{Make sure to enter the correct
  conference title from your rights confirmation email}{June 03--05,
  2018}{Woodstock, NY}
\acmISBN{978-1-4503-XXXX-X/2018/06}

\author{Ansel Kaplan Erol}
\email{aerol3@gatech.edu}
\affiliation{%
   \institution{Georgia Institute of Technology}
   \city{Atlanta}
   \state{Georgia}
   \country{USA}
}
\author{Divya Mahajan}
\email{divya.mahajan@gatech.edu}
\affiliation{%
   \institution{Georgia Institute of Technology}
   \city{Atlanta}
   \state{Georgia}
   \country{USA}
}

\begin{document}

\title{\helios: Decentralized Scheduling for Hardware-Aware Orbital Intelligence}

\renewcommand{\shortauthors}{A. Erol \& D. Mahajan}


\begin{abstract}
Earth-observation satellites are emerging as distributed edge computing platforms for time-critical tasks such as disaster response and maritime intelligence, carrying onboard accelerators that enable in-orbit inference. Scheduling these workloads fundamentally differs from datacenter scheduling: energy is harvested intermittently, and decisions are temporally coupled where processing an image eagerly can deplete battery reserves needed minutes later for higher-value events. Existing schedulers rely on static priorities or resource thresholds and lack mechanisms to adaptively shed lower-value work as constraints tighten.

We present \helios, a lightweight, decentralized runtime for the strict constraints of low-power orbital edge systems. \helios\ enables adaptive scheduling under complex conditions by compressing heterogeneous and time-varying hardware constraints such as battery charge, thermal headroom, and queue backlog into a single state-dependent marginal cost of execution. Derived from a barrier function that rises sharply as system state approaches safety limits, this cost encodes both instantaneous resource pressure and future risk. 
This local cost signal also serves as a constellation-wide coordination primitive. Tasks execute only when their value exceeds the current cost, producing value-ordered load shedding without explicit prioritization policies. Furthermore, if local execution costs exceed those of a neighbor, tasks are dynamically offloaded over inter-satellite links. This achieves distributed load balancing without routing protocols, queue exchange, or global state. We show this single signal is sufficient to align scheduling decisions across a disconnected, resource-constrained system.

We evaluate \helios\ using a hardware-validated, multi-day simulation of a 143-satellite constellation, grounded in measurements from a physical Jetson Orin Nano Super. \helios\ improves scientific goodput by 20\% and image-processing throughput by 31\% over priority-based scheduling, while maintaining 2.2$\times$ higher mean battery reserves. Under increasing demand, \helios\ gracefully sheds work to achieve 5.2$\times$ the execution rate of static scheduling rather than collapsing under contention.
\end{abstract}


\settopmatter{printfolios=true}
\maketitle
\pagestyle{plain}

\section{Introduction}


\begin{sloppypar} Modern Low-Earth Orbit (LEO) constellations leverage machine learning accelerators capable of running models for wildfire detection, flood mapping, and maritime surveillance directly in orbit, eliminating the multi-hour latency of downlink-then-analyze pipelines \cite{orbitaledge, spire2025wildfiresat, pelican}. The energy budget powering these inferences, however, is physically coupled to sunlight. Satellites spend roughly 35\% of each orbit in Earth's shadow, known as eclipse, operating entirely on a 50--100 Wh battery \cite{powerbudgets}. A scheduler that processes at full rate during sunlight risks exhausting those reserves before eclipse ends, which could leave no resources available at the moment the satellite passes over a region of interest. Satellites image continuously, but onboard compute can service only a fraction of arriving events; the scheduler must triage, deciding which images are worth processing at all, not merely in what order.
\end{sloppypar}

A satellite's onboard compute is fixed at launch within a rigid mass and power envelope, motivating efficient edge accelerators like the NVIDIA Jetson~\cite{jetsonorin} and Google Edge TPU~\cite{coral} for large-scale industry OEC initiatives~\cite{pelican, aguera2025suncatcher}. Compute capacity cannot be provisioned on demand. Energy arrives only during sunlit phases and cannot be replenished. These properties create an \emph{inter-temporal coupling}: processing an image now consumes energy that may be needed minutes later, during eclipse, for a detection of far greater value. Delaying or dropping work risks missing time-critical phenomena. Figure \ref{fig:intro} illustrates this challenge.

\begin{figure*}
\includegraphics[width=0.8\linewidth]{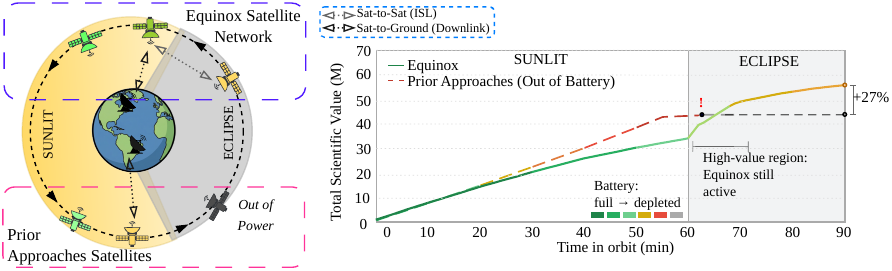}
\caption{Scheduling for OEC satellites must maintain processing power for the 30\% of orbit spent in eclipse.}
\label{fig:intro}
\Description{Diagram showing Earth and atmosphere with satellites executing tasks in orbit. Above, a bar depicts the Nexus runtime as a central box mediating between incoming tasks on one side and orbital hardware constraints on the other.}
\end{figure*}

The scheduling problem is further complicated by the infeasibility of centralized control. A satellite is visible to a ground station for only minutes per orbit. For the remainder, including most of eclipse, it operates entirely without contact, making it impractical to delegate scheduling decisions to the ground. The runtime must act from local hardware state alone, with no global view and no coordinator. This combination of non-replenishable energy, intertemporal coupling, and enforced autonomy places satellite scheduling closer to real-time control under physical resource limits than to traditional workload management. 

Existing approaches address only one dimension of this problem. Value-based systems \cite{serval, erol2025earthsight} rank the tasks by expected utility, thus prioritizing certain tasks before others but ignoring hardware state. 
These systems greedily excecute the highest priority work available at time $t$. However, by ignoring future resource constraints, they risk depleting battery reserves needed for higher-value inference tasks that arrive later, resulting in suboptimal performance over the full orbit.
%
%
Resource-based systems \cite{liu2024phoenix, li2025orbitchain} enforce battery or thermal limits and may redistribute load via inter-satellite links (ISLs), but treat all tasks as uniform importance, lacking a mechanism to preferentially retain higher-value work under contention. \textit{Thus, no existing system jointly accounts for task value and time-varying resource constraints in a decentralized setting.}

We present \helios, a decentralized runtime for resource-constrained orbital edge systems. Solutions that reason over the full resource-task state space exceed the computational budget of satellites already saturated by inference, and ground-directed systems cannot adapt in time as contact windows span only minutes per orbit. Our key insight for overcoming this challenge is that heterogeneous and time-varying hardware constraints such as battery charge, thermal headroom, and queue backlog, can be compressed into a single \emph{state-dependent marginal cost of execution}. This cost is derived from a barrier function that rises sharply as system state approaches safety limits, encoding both instantaneous resource pressure and the risk of future depletion. Tasks execute only when their value exceeds this cost, producing value-ordered load shedding without explicit priority policies.

This formulation yields two important properties. First, it induces \emph{staggered dropout}: as resources tighten, tasks are rejected in strict order of increasing value, leaving the highest-value work as the last to lose access. Second, the cost signal also serves as a coordination primitive across satellites. When the cost of execution on the local node exceeds that of a neighbor (plus communication cost), tasks are offloaded over inter-satellite links. This enables distributed load balancing without routing protocols, queue exchange, or global state; the cost signal is sufficient to guide both scheduling and placement decisions.

\helios\ is the first decentralized runtime to unify value-aware scheduling with state-dependent resource control in orbital edge systems. Our contributions are as follows:

\begin{itemize}[nosep,leftmargin=*]
\item We introduce a scheduling framework that compresses battery, thermal, and queue state into a single state-dependent marginal cost, enabling value-aligned execution without relying solely on task priorities or resource accounting.
\item We design a barrier-based cost function that enforces safety constraints by construction and induces value-ordered task rejection under arbitrary arrival sequences.
\item We show that this cost signal serves as a unified coordination primitive, enabling inter-satellite load balancing without routing protocols or global state.
\item We propose a constellation-level evaluation framework spanning scientific value and goodput, event detection coverage, time with battery reserves , brownout risk, and scientific load balance.
\end{itemize}

We evaluate \helios\ via a 143-satellite simulation using the fMoW dataset \cite{christie2018fmow}, anchored by hardware profiling on an NVIDIA Jetson Orin \cite{jetsonorin}. Our central finding is that strategically shedding lower-value tasks to preserve battery headroom directly maximizes long-term scientific throughput, particularly during energy-starved eclipse phases where greedy schedulers predictably fail.
By sustaining operations through eclipse transitions, \helios\ extracts 20\% higher scientific value per hour than priority baselines (+64\% vs. FIFO) and processes 31\% more images, all while maintaining 2.2$\times$ greater mean battery reserves. Under severe contention (16 tasks), \helios\ sustains 5.2$\times$ the execution rate of static scheduling by shedding marginal work rather than collapsing. Furthermore, we show \helios's marginal cost formulation allows operators the flexibility to shift \helios\ along the scientific throughput--battery health trade-off to match mission requirements.\footnote{Anonymized code: \url{https://github.com/equinox-authors/equinox}}
\section{Background}
\label{sec:background}

\subsection{Orbital Edge Computing}

Early Earth observation platforms such as Sentinel and Aqua transmitted raw sensor data to ground stations for offline processing \cite{sentinel_data, aqua_cer_ssf}. This \textit{bent-pipe} architecture is fundamentally constrained by short contact windows (typically under 10 minutes per orbit), shared downlink bandwidth, and the sheer volume of imagery generated by modern constellations \cite{l2d2, kodan}. As a result, time-critical events such as wildfires, floods, and maritime incidents cannot be processed in time when relying on ground-side analysis.

Modern constellations address this limitation by moving inference onboard the satellite. Nanosatellites equipped with commercial off-the-shelf accelerators \cite{spireconstellation, planetconstellations, jetsonorin} execute models in orbit and transmit only detections or compressed summaries. This paradigm, Orbital Edge Computing (OEC)~\cite{orbitaledge}, reduces downlink demand and enables real-time response at constellation scale. Commercial deployments are already underway: Spire Global performs in-orbit wildfire detection \cite{spire2025wildfiresat}, Planet integrates NVIDIA Jetson accelerators into its Pelican satellites \cite{pelican}, and Google's Project Suncatcher explores general-purpose orbital inference \cite{aguera2025suncatcher}.

\begin{figure}
\includegraphics[width=0.7\columnwidth]{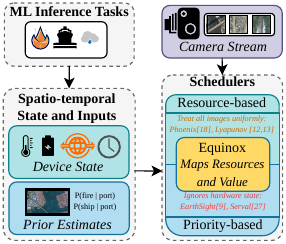}
\caption{Context of \helios\ to Earth-observation alternatives. Only \helios\ incorporates both hardware state and image priors to schedule Earth Observation workloads.}
\label{fig:comparison}
\vspace{-2ex}
\end{figure}

\subsection{The Orbital Scheduling Problem}

Moving inference in-orbit replaces the downlink bottleneck with a tightly coupled resource management problem. Unlike datacenters, both compute and energy are severely constrained and evolve over time. Energy availability follows the orbital cycle, increasing during sunlight and depleting during eclipse, while compute demand varies with the spatiotemporal distribution of incoming imagery and task complexity. As a result, resource availability and workload demand are continuously changing, creating a scheduling problem with strong intertemporal dependencies.

\niparagraph{Compute oversubscription.}
Onboard accelerators are orders of magnitude weaker than datacenter hardware. Meanwhile, high-resolution sensors generate imagery at rates that exceed what can be processed in real time, creating a persistent compute deficit. As a result, the scheduler must continuously decide which images to process, which models to apply, and which tasks to defer or drop.

\niparagraph{Intermittent energy and eclipse.}
Satellites harvest solar power during sunlight for roughly 65\% of each orbit and rely on batteries for the remaining 35\% \cite{liu2024phoenix} of the orbit during Eclipse. With only 50--100\,Wh \cite{powerbudgets}, yet decisions that need to be made are intertemporal: executing a task now will consume energy that may be required minutes later during eclipse. A scheduler that runs aggressively during sunlight risks depleting reserves and becoming unavailable when high-value events occur. Maintaining sufficient battery headroom is therefore essential not only for safety, but for sustained throughput.

\niparagraph{Thermal limits.}
Thermal constraints further restrict sustained computation. In the absence of convective cooling, heat must be dissipated via radiation, and continuous inference drives accelerators toward thermal throttling limits. Thermal state is further coupled with the solar cycle: sunlit satellites absorb external radiation while generating compute heat, reducing available headroom.

\niparagraph{Scientific value heterogeneity.}
Workloads are highly heterogeneous in value. An image over open ocean may carry negligible utility, while one over a dense port or wildfire region may be highly valuable. Uniformly allocating compute under scarcity wastes resources on low-value inputs and risks missing high-value events when constraints tighten.

\niparagraph{Orbital dynamics and spatial heterogeneity.}
Beyond per-satellite constraints, orbital scheduling is inherently distributed.
Resource availability and workload demand are not only time-varying but also spatially distributed across the constellation. At any given time, satellites occupy different phases of the orbital cycle: some are in sunlight with abundant energy, while others are in eclipse operating on limited battery reserves. Similarly, observation opportunities are unevenly distributed across the Earth's surface, with crucial events concentrated in specific regions and at specific times.

This spatial and temporal heterogeneity creates opportunities for cooperation across satellites. A satellite in eclipse may be energy-constrained while a neighboring satellite in sunlight has surplus capacity. However, exploiting this asymmetry requires transferring tasks across inter-satellite links (ISLs), introducing additional latency and energy costs. Scheduling decisions must therefore account not only for local resource state, but also for the relative resource availability of neighboring nodes and the cost of offloading work.

\subsection{The Gap in Existing Approaches}

These constraints create a scheduling problem that existing systems address only partially (Figure \ref{fig:comparison}). Scientific value-aware systems \cite{serval, erol2025earthsight} prioritize tasks by expected utility, but ignore resource state, executing moderate-value work even when doing so consumes energy needed for higher-value events arriving shortly thereafter. Resource-aware systems \cite{kodan, hu2023lyapunov, zhang2024peer, li2025orbitchain, liu2024phoenix} enforce battery or thermal limits and may redistribute load via inter-satellite links (ISLs), but treat tasks uniformly, lacking a mechanism to preferentially retain higher-value work under contention.

ISLs provide a mechanism to exploit this spatial heterogeneity by redistributing work across satellites. However, they introduce an additional layer of complexity: deciding not only whether a satellite has spare capacity, but \emph{whether a specific task should be offloaded}. This requires jointly reasoning about local resource pressure, the value of the task, the neighbor's available capacity, and the cost of transfer. Existing approaches capture either capacity or value, but not both, and therefore cannot make value-aware offloading decisions under resource constraints.

What is missing is a unified abstraction that encodes system state in a form that can drive both scheduling and routing decisions. Such an abstraction must reflect current resource pressure, anticipate future constraints, and allow tasks to be selected based on their value relative to system state. \helios\ provides this abstraction through a state-dependent marginal cost signal that enables both value-aware scheduling and decentralized coordination.

\section{\helios\ System}
\label{sec:system}

\subsection{Overview}
\label{sec:system:overview}

\helios\ is a decentralized runtime for distributed resource management in orbital edge systems. It is motivated by a key challenge: that scheduling decisions in an OEC setting must simultaneously account for time-varying resource constraints, heterogeneous scientific task values, and distributed execution across satellites. However, existing approaches treat these concerns in isolation.

\textbf{\textit{\helios\ addresses this challenge through a single unifying abstraction: a state-dependent marginal execution cost that captures the instantaneous pressure on a satellite’s resources.}}
Rather than separating prioritization, scheduling, and routing into distinct mechanisms, \helios~reduces all decisions both local and global, to a single comparison: task value versus marginal execution cost. This simple, low-dimensional signal is critical for edge deployment. Scheduling algorithms that optimize over the full resource-task state space are computationally out of reach for satellite-class hardware already saturated by inference workloads. Ground-directed control is equally impractical: contact windows span only minutes per orbit, making reactive teleoperation too slow to adapt to eclipse transitions or burst arrivals. A single closed-form scalar computed from local sensor readings provides the necessary control without competing with inference for CPU headroom. 

This signal compresses battery state, thermal headroom, and queue congestion into a single control scalar that reflects current utilization and proximity to critical limits. Thus, \helios\ selects tasks based on their predicted scientific value relative to resource pressure, rather than through separate priority and resource management mechanisms.
Each satellite independently computes this marginal cost and uses it as a local control signal. There is no centralized controller, and coordination emerges implicitly through comparisons of cost across satellites.

As shown in Figure~\ref{fig:nexus_architecture}, the \helios\ runtime continuously executes a four-step loop: (1) computes the marginal execution cost from local sensor readings; (2) evaluates incoming tasks to estimate their expected scientific value, based on event likelihood and importance weighting; (3) executes tasks whose expected science value exceeds the local cost; and (4) compares costs with neighbors via inter-satellite links (ISLs) to offload tasks to cheaper nodes, dropping or deferring the rest.
By strictly decoupling tasks (which act as hardware-agnostic value estimators) from the system (which exposes only the execution cost), \helios\ replaces explicit priority queues, hand-tuned thresholds, and resource-specific policies with a unified, elegant decision rule.


\begin{figure*}
    \centering
    \includegraphics[width=0.92\linewidth]{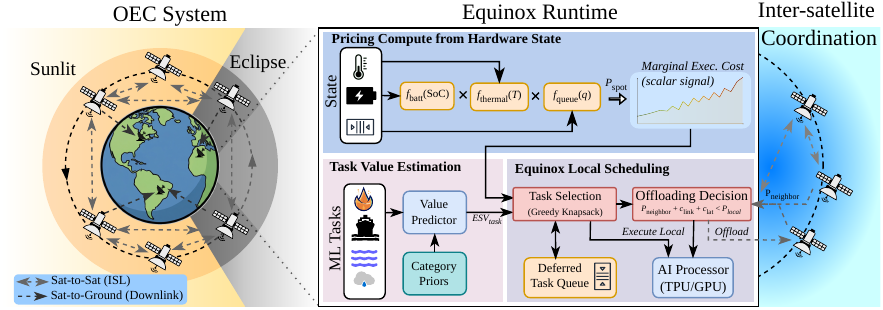}
    \caption{Overview of the \helios~framework, illustrating per-satellite marginal execution cost, per-satellite local scheduling, and ISL arbitrage components within a decentralized runtime for distributed resource management in orbital edge systems.}
    \label{fig:nexus_architecture}
\Description{Architecture diagram showing per-satellite spot pricing, task bidding, greedy knapsack clearing, and ISL arbitrage components.}
\end{figure*}

\subsection{Marginal Execution Cost}
\label{sec:system:pricing}

\helios\ represents system state of each satellite through the marginal execution cost $P$, defined as the incremental cost of processing one additional task at the current hardware state. This signal unifies multiple resource constraints into a single control variable that can drive both scheduling and routing decisions.
The marginal cost increases as any resource approaches its limit, encoding not only current utilization but also proximity to unsafe or unsustainable operating regimes. Tasks are admitted when their expected value exceeds this cost, and rejected or deferred otherwise.

\helios\ decomposes marginal execution cost into independent factors: \textbf{battery state, thermal headroom, and queue congestion.} Each factor captures how close the system is to a limiting condition along a single resource dimension.
We combine these factors multiplicatively so that pressure on any one resource can dominate the overall cost. This ensures that no single resource (e.g., a healthy battery) can mask risk in another (e.g., thermal limits), and that the cost rises sharply as the system approaches any critical boundary. The resulting execution cost formulated by \helios\ is:

\begin{equation}\footnotesize
  P = P_{\text{base}} \cdot f_{\text{batt}}(\text{SoC}) \cdot f_{\text{thermal}}(T) \cdot f_{\text{queue}}(q)
  \label{eq:spot_price}
\end{equation}

\noindent where $P_{\text{base}}$ is the nominal cost under standard conditions. Calibration of coefficients and hyperparameters is deferred to Section~\ref{sec:eval:setup} and Appendix~\ref{app:calibration}. Importantly, the sensitivity coefficients $\beta$ and $\gamma_q$ are operational knobs: increasing $\beta$ enforces stricter energy conservation, improving battery reserve time at the cost of throughput, while smaller values favor processing under tighter energy margins. This allows operators to tune \helios\ along the scientific throughput--battery health trade-off to match mission priorities without modifying the control structure (Figure~\ref{fig:pareto}, Table~\ref{tab:all_results}).

\niparagraph{Battery factor.}
Battery state governs the system’s ability to survive future execution, particularly across eclipse. Unlike other resources, energy is not continuously replenished, and mistakes are irreversible over the short horizon of an orbit. The battery factor therefore encodes proximity to energy depletion as a hard constraint.
We model this using a barrier function around the critical state of charge $\text{SoC}_c$.
\begin{equation}\footnotesize
  f_{\text{batt}}(\text{SoC}) = 1 + \frac{\beta}{(\text{SoC} - \text{SoC}_c)^2}
  \label{eq:f_batt}
\end{equation}

As $\text{SoC} \rightarrow \text{SoC}_c^+$, the marginal cost diverges, forcing the system to shed load before running out of power. This captures a key asymmetry: while tasks can be deferred, energy shortfall during eclipse cannot be corrected in real time. We calibrate $\beta$ such that this behavior activates early enough to preserve power for an eclipse interval.

\niparagraph{Thermal factor.}
Thermal state limits how long computation can be sustained. Unlike battery constraints, which are driven by future availability, thermal limits arise from instantaneous power density and the inability to dissipate heat efficiently in orbit.
We encode this using a soft barrier as temperature approaches the throttle threshold $T_{\max}$:
\begin{equation}\footnotesize
  f_{\text{thermal}}(T) = 1 + \frac{\gamma_T \cdot (T - T_0)^2}{(T_{\max} - T)^2 + 1}
  \label{eq:f_thermal}
\end{equation}

where $T_0$ is the nominal operating temperature, $T_{\max}$ is the thermal throttle threshold, and $\gamma_T$ is a sensitivity coefficient controlling how steeply cost rises with thermal headroom loss. This factor remains near one under nominal conditions, but rises sharply as thermal headroom collapses. Its role is preventive, avoiding schedules that appear feasible in the short term but would trigger throttling and reduce effective throughput, especially during sunlit phases when both external heating and compute load are high.

\niparagraph{Queue factor.}
Queue factor captures mismatch between input rate and processing capacity. While battery and thermal factors are physical limits, the queue factor captures system-level congestion.
We model this as a linear function of deferred workload:
\begin{equation}\footnotesize
  f_{\text{queue}}(q) = 1 + \gamma_q \cdot q
  \label{eq:f_queue}
\end{equation}

\noindent where $q$ is the number of images in the deferred queue and $\gamma_q$ is the congestion sensitivity coefficient.
As backlog grows, the marginal cost increases, biasing the system towards higher-value tasks and preventing accumulation of low-value work. Unlike the barrier functions, this term provides continuous back-pressure rather than enforcing a hard limit.
$P$ represents the current cost of executing one additional task. Tasks with expected value above $P$ are admitted, while those below $P$ are deferred or dropped.

\subsection{Per-Satellite Scheduling and Resource Allocation}
\label{sec:system:bidding}

\helios\ operates as a continuous decision process in which tasks compete for execution based on their expected value relative to current system cost. Rather than separating prioritization, scheduling, and routing into distinct mechanisms, the system reduces all decisions to a single comparison: task value versus marginal execution cost.

\niparagraph{Value representation.}
Each task $t$ corresponds to a detection objective and expresses the expected scientific value of processing an image as a scalar:
\begin{equation}\footnotesize
  ESV_{t} = P(\text{event}_t \mid \text{category}) \cdot \text{acc}_t \cdot w_t
  \label{eq:esv}
\end{equation}
where $P(\text{event} \mid \text{category})$ captures the likelihood of relevance to an image's region category, $\text{acc}_t$ reflects model capability, and $w_t$ encodes the scientific importance of task $t$. This formulation separates opportunity, capability, and importance, allowing heterogeneous tasks to be compared on a common scale without exposing its internals.

\niparagraph{Admission and Resource-constrained Selection.}
At runtime, tasks are admitted only if their value exceeds the current execution cost $P$. This condition acts as a global, state-dependent threshold: as resource pressure increases, fewer tasks remain economically viable. Crucially, this replaces explicit priority queues and hand-tuned policies with a single, consistent decision rule that adapts automatically to system conditions. Admitted tasks compete for limited compute and energy resources. Because the system operates under persistent oversubscription, selection is unavoidable even in nominal conditions. The runtime therefore chooses a subset of tasks that maximizes total expected value under the instantaneous compute budget:
\begin{equation}\footnotesize
  \max \sum_{i \in S} ESV_i \quad \text{s.t.} \quad \sum_{i \in S} g_i \leq G
  \label{eq:knapsack}
\end{equation}
where $g_i$ denotes the compute demand of task $i$. We approximate this using a value-density ordering $V_i / g_i$, reflecting the system’s objective of extracting the most value per unit of constrained resource. 
The marginal cost $P$ plays a critical upstream role: by filtering out low-value tasks before selection, it reduces problem complexity and ensures that the remaining competition is focused on high-impact work.

\niparagraph{Temporal adaptation.}
Tasks that are not selected are not immediately discarded. Instead, they are deferred and reconsidered as system conditions evolve. Because $P$ varies over time with battery, thermal, and queue state, tasks that are uneconomical under current conditions may become viable later, for example, after exiting eclipse or clearing congestion.
This mechanism effectively shifts demand across time without requiring tasks to predict future conditions. The system, rather than the task, absorbs temporal uncertainty.

\niparagraph{Graceful degradation.}
\label{sec:system:dropout}
As resource pressure increases, the system does not fail abruptly. Instead, rising marginal cost progressively filters out lower-value tasks while preserving higher-value ones. 
This behavior can be characterized analytically. A task $a$ ceases to execute when $P > V_a$, which defines a critical state of charge $\text{SoC}_a^*$:
\begin{equation}\footnotesize
  \text{SoC}_a^* = \text{SoC}_c + \sqrt{\frac{\beta \cdot P_{\text{base}}}{ESV_{task} - P_{\text{base}}}}
  \label{eq:dropout_soc}
\end{equation}
Because $\text{SoC}_a^*$ decreases monotonically with $V_a$, tasks exit the system in strict value order. High-value tasks persist deeper into constrained regimes, while lower-value tasks are shed earlier. 
This value-ordered degradation is not explicitly programmed; it directly emerges from the cost formulation, ensuring that the system remains aligned with scientific priorities even under severe resource constraints.

\subsection{Satellite Constellation Load Redistribution}
\label{sec:system:isl}
Orbital dynamics introduce persistent asymmetries in cost. Satellites in eclipse experience rising marginal cost due to energy constraints, while sunlit satellites operate at lower cost. 
Thus, the same value--cost abstraction extends naturally across satellites. Each node exposes its marginal execution cost, and tasks are routed toward locations where they can be executed more efficiently.

A task is offloaded when a neighboring satellite offers lower effective cost after accounting for transfer overhead:
\begin{equation}\footnotesize
  P_{neighbor} + c_{\text{link}} + c_{\text{lat}} < P_{local}
\end{equation}  
\noindent where $c_{\text{link}}$ is the ISL transmission cost and $c_{\text{lat}}$ is a penalty for routing latency. 
As a result, load redistribution is a process of cost equalization across the constellation, rather than explicit load balancing. This process is entirely decentralized: no satellite maintains a global view or explicitly balances load.
These gradients naturally drive a flow of tasks toward energy-rich regions, causing workload to follow the solar cycle without explicit coordination or eclipse-aware logic.

Routing decisions rely only on observable hardware state and exclude queue congestion that is not live across nodes:
\begin{equation}
P_n^{\text{phys}} = P_{\text{base}} \cdot f_{\text{batt}}(SoC_{neighbor}) \cdot f_{\text{therm}}(T_{neighbor}).
\end{equation}
This ensures that offloading decisions remain stable and grounded in shared physical signals. Because routing is governed by the same value–cost comparison as local scheduling, it preserves value awareness: only tasks that remain worthwhile after transfer are offloaded. Unlike capacity-based schemes that redistribute load uniformly, \helios\ \textit{preserves value-awareness during routing.} Offloading is only beneficial if the task remains worth executing at the destination after accounting for transfer costs. As a result, high-value tasks are preferentially redistributed under contention, while low-value tasks are filtered out locally.

Overall, \helios~extends the scheduling problem from a single node to a distributed system. Locally, scheduling selects tasks that maximize value under resource constraints; globally, it determines where in the constellation that value should be realized. The marginal execution cost serves as a common currency across satellites, allowing each node to make independent decisions that collectively approximate a system-wide allocation of work. In this view, offloading is not a separate mechanism but a continuation of scheduling across nodes: tasks are placed where their value can be realized at lowest cost.

\subsection{Satellite Constellation Performance Metrics}
\label{sec:system:metrics}

\helios\ makes decisions using a scalar control signal, the marginal execution cost, but evaluating its behavior requires a broader view. No single metric can capture the trade-offs inherent in orbital systems, where value, resource availability, and hardware health are tightly coupled.
Standard performance metrics such as throughput are insufficient: they ignore resource constraints, system safety, and the value of observations. A throughput-maximizing scheduler may exhaust battery reserves and miss critical events during eclipse, while an overly selective system may achieve high per-image quality but fail to capture rare events. 
We therefore evaluate how a simple, low-dimensional control signal translates into system-level behavior using a set of complementary metrics spanning value, efficiency, and operational readiness.

\niparagraph{Scientific Value and Goodput.}
Scientific Value (SV) captures the total realized scientific utility, i.e., the weighted and summed value of detected events:
\begin{equation}\footnotesize
  \text{SV} = \sum_{i \in \mathcal{E}} \sum_{t \in \mathcal{T}_i}
  \mathbf{1}[\text{event}_a \in \text{content}(i)] \cdot w_t
  \label{eq:sv}
\end{equation}
\noindent where $\mathcal{E}$ is the set of executed images and $\mathcal{T}_i$ is the set of tasks applied to image $i$. SV weights detections by their scientific importance; for example, fire detection is more time-critical than routine environmental monitoring. Task-specific weights are provided in Section~\ref{sec:eval:setup}.

Moreover, Scientific Goodput (SG) measures the rate at which scientific value is produced, i.e., scientific value per hour. Unlike raw throughput, the Scientific Goodput rewards systems that sustain high-value processing over time, and therefore captures both selectivity and sustained execution.

\niparagraph{Event Coverage.}
Event Coverage measures the fraction of observable events that are actually detected. This metric ensures that selectivity does not come at the cost of missing rare but important events. A system with high image throughput but no value-aware selection may achieve high coverage even with low per-image utility, while a highly selective system may achieve strong per-image utility but low overall coverage if it processes too few images.

\niparagraph{Scientific Load Balance.}
Scientific Load Balance measures how evenly scientific value is distributed across each satellite in the constellation:
\begin{equation}\footnotesize
  \text{LoadBalance} = \left(1 - \frac{\sum_{i=1}^{N}\sum_{j=1}^{N} |v_i - v_j|}{2N \sum_{i=1}^{N} v_i}\right)\times 100\%
  \label{eq:loadbalance}
\end{equation}
where $v_i$ is the total Scientific Value accumulated by satellite $i$ over the simulation. A value of $100\%$ indicates perfectly uniform distribution across satellites, while $0\%$ indicates complete concentration in a single satellite. Higher values are better. This formula is based on the Gini index \cite{gini1921measurement}.

\niparagraph{Battery Reserve Time and Brownout Risk.}
Battery Reserve Time captures operational readiness: the fraction of satellite-timesteps for which battery is above the safe operating threshold
$\text{SoC}_{\text{safe}} = 0.35$. Satellites above this threshold can perform unexpected high-value events without immediate energy stress.

\niparagraph{Brownout Risk.} Brownout Risk measures the fraction of satellite timesteps in which battery level falls below the near-critical  threshold:
$\text{SoC}_{\text{brn}} = 0.20$. Lower brownout risk is better. In this regime, a satellite has little usable reserve above the hard cutoff ($\text{SoC}_c = 0.15$), making it vulnerable to missed detections and unable to absorb additional load.

\textit{Together, these metrics evaluate the system along three axes: value (Scientific Value and Goodput, Event Coverage), operational health (SG, Brownout Risk), and constellation-wide scientific load balance (Scientific Load Balance).}
\section{Evaluation}
\label{sec:eval}

\subsection{Experimental Setup}
\label{sec:eval:setup}

We evaluate \helios\ with a hardware-augmented satellite constellation simulator that combines validated orbital dynamics with measurements from a real edge platform. Rather than relying on abstract resource models alone, we profile a physical Jetson Orin (Section~\ref{sec:eval:hw}) and use those measurements to instantiate the simulator’s compute, energy, and physical parameters (battery and thermals). 

\niparagraph{Constellation setting.} The orbital environment is modeled using standard, validated astrodynamics components: SGP4 for orbit propagation~\cite{vallado2006}, Skyfield for eclipse geometry~\cite{rhodes_skyfield}, and SciPy for pairwise ISL distance computation~\cite{scipy2020}, all widely used in prior related work~\cite{orbitaledge, liu2024phoenix}. 
The overall simulation encompasses 143 LEO satellites arranged in 13 orbital planes of 11 satellites each, at 500\,km altitude, with an orbital period of approximately 94 minutes. ISL topology is recomputed each step; satellites within 5{,}000\,km are considered linked, corresponding to demonstrated optical terminal capabilities in LEO mega-constellations \cite{Chaudhry2021} and ensuring 8--12 neighbors per satellite on average. To quantify robustness beyond the homogeneous Walker baseline, we also run orbital heterogeneity stress tests (phase jitter, altitude tiers, and their combination), reported in Appendix~\ref{app:orbital_heterogeneity}.

\sloppy
\niparagraph{Satellite configuration.}
Each satellite comprises a 100\,Wh lithium-ion battery, a 120\,W peak solar panel (cosine-attenuated), and an NVIDIA Jetson AGX Orin compute module (Table~\ref{tab:hw_config}). The compute budget is 2.0\,GFLOPs/s, calibrated to the power-envelope-limited throughput: at 10\,W per concurrent image (heavy tier, 1.8\,GFLOPs), the budget supports roughly one image per scheduling step, resulting in compute demand $1.35\times$ supply. Simulations use a timestep of $\Delta= 1s$ over a 72-hour duration.

\niparagraph{Workload.} Each satellite captures 90 images per minute, producing approximately 33\% compute oversubscription relative to the profiled Orin-class power envelope. Event probabilities follow the 1{,}047{,}691-image, 62 land-use category fMoW dataset~\cite{christie2018fmow} (e.g., ports have high vessel probability, while agricultural regions have near-zero probability), providing the spatial and temporal heterogeneity that \helios\ is designed to exploit. We evaluate four detection tasks---fire, flood, vessel, and environmental monitoring---with base event probabilities of 0.05, 0.08, 0.12, and 0.10, respectively. We train EfficientNet~\cite{tan2020efficientnetrethinkingmodelscaling} models, following prior work~\cite{erol2025earthsight, capogrosso2026tinymlenhancescubesatmission}, across three model sizes (B0, B2, B3) with the most accurate models achieving 93-95\% precision and recall. We assign scientific weights $w = \{200, 100, 50, 20\}$ to reflect relative environmental urgency. This produces well-separated Expected Science Values that result in staggered task dropout. Our four-task workload already exceeds the 1-3 task standard in related work~\cite{liu2024phoenix, kodan, serval, orbitaledge, li2025orbitchain} and we demonstrate scale-up to sixteen tasks per image in Section \ref{sec:eval:scaling}.



\begin{table}[t]
  \centering
  \caption{NVIDIA Jetson AGX Orin hardware configuration.}
  \vspace{-2ex}
  \label{tab:hw_config}
  \footnotesize
  \begin{tabular}{ll}
    \toprule
    \textbf{Component} & \textbf{Specification} \\
    \midrule
    CPU    & 12-core Arm Cortex-A78AE v8.2, 64-bit \\
    GPU    & 2048-core NVIDIA Ampere + 64 Tensor Cores \\
    AI perf.\ & 275 TOPS (CPU, GPU, and DLA combined) \\
    Memory & 32\,GB LPDDR5 (204.8\,GB/s bandwidth) \\
    Power  & 15\,W (idle) -- 60\,W (max TDP) \\
    \bottomrule
    \vspace{-3ex}
  \end{tabular}
\end{table}

\niparagraph{Baselines.}
\label{sec:eval:baselines}
We compare \helios\ with four baselines representing progressively richer scheduling policies. \textbf{Static FIFO} processes images strictly in arrival order, with no notion of task value. \textbf{Priority Queue} ranks images by expected scientific value and executes them in descending order. This baseline incorporates value awareness but remains oblivious to hardware state. To avoid pathological behavior at very low battery levels, we add a simple guardrail that suppresses tasks with scientific weight $w < 100$ below 20\% battery level.
\textbf{ESA}~\cite{huang2013esa} is a Lyapunov-based online optimizer that balances immediate reward against a virtual battery backlog to maintain long-term energy stability.
\textbf{Phoenix}~\cite{liu2024phoenix} uses centralized coordination to offload tasks from eclipsed satellites to the highest-SoC neighbor in sunlit orbital planes. Tasks are deferred locally if their TTL expires before the eclipse ends; otherwise, they are offloaded.
Finally, we evaluate three variants of \helios. \textit{\helios\ (no-ISL)} disables inter-satellite offloading. \textit{\helios\ (no-context)} replaces category-conditioned ESVs with a uniform base-rate, isolating the value of contextual priors. \textit{\helios\ (noisy-context)} perturbs those priors with log-normal noise ($\sigma=0.25$), evaluating robustness to imperfect context.

\niparagraph{Hardware thresholds.}
The hardware thresholds come from the device itself: $\text{SoC}_c = 0.15$ is the minimum battery reserve, $\text{SoC}_{\text{safe}} = 0.35$ is the point where the scheduler begins conserving energy, $T_0 = 50^\circ$C is the nominal temperature, and $T_{\max} = 85^\circ$C is the throttling threshold. We then choose the sensitivity coefficients so that the scheduler behaves intuitively: all tasks remain admissible under normal conditions, lower-value tasks are dropped first as eclipse approaches, and thermal and queue effects stay small unless the system is actually under stress. 

\subsection{Main Results}
\label{sec:eval:main}

\begin{table*}[t]
  \centering
  \caption{72-hour fMoW results over a 143-satellite constellation. \helios\ delivers the best overall trade-off between scientific return and operational resilience, achieving the highest scientific goodput and event coverage while maintaining stronger battery health and more even load distribution than the baselines.}
  \label{tab:all_results}
  \setlength{\tabcolsep}{5pt}\footnotesize
  \begin{tabular}{lccc|ccc|c}
    \hline
    & \multicolumn{3}{c}{\textbf{\rule{0pt}{2.5ex}Scientific Value}} & \multicolumn{3}{c}{\textbf{\rule{0pt}{2.5ex}Operational Health}} & \multicolumn{1}{c}{\textbf{\rule{0pt}{2.5ex}Fairness}} \\
    \cline{2-8}
    \textbf{\rule{0pt}{3ex}System\rule[-1.5ex]{0pt}{0pt}} & \textbf{\shortstack{\rule{0pt}{2.5ex}Scientific Goodput\\(M~SV/hr)}} & \textbf{\shortstack{\rule{0pt}{2.5ex}Throughput\\(K\,Images/hr)}} & \textbf{\shortstack{\rule{0pt}{2.5ex}Event \\Coverage (\%)}} & \textbf{\shortstack{\rule{0pt}{2.5ex}Mean Battery\\Level (\%)}} & \textbf{\shortstack{\rule{0pt}{2.5ex}Battery Reserve\\Time (\%)}} & \textbf{\shortstack{\rule{0pt}{2.5ex}Brownout\\Risk (\%)}} & \textbf{\shortstack{\rule{0pt}{2.5ex}Science\ Load\\Balance (\%)}} \\
    \hline
    Static                   & \cellcolor{red!20} 8.20 & \cellcolor{red!20} 346 & \cellcolor{red!20} 44.8 & \cellcolor{red!15} 20.4 & \cellcolor{red!19} 7.4 & \cellcolor{red!10} 64.4 & \cellcolor{red!10} 91.2 \\
    Phoenix \cite{liu2024phoenix}                 & \cellcolor{red!13} 9.00 & \cellcolor{red!7} 380 & \cellcolor{red!11} 49.1 & \cellcolor{green!20} 84.8 & \cellcolor{green!20} 99.9 & \cellcolor{green!20} 0.0 & \cellcolor{red!20} 90.2 \\
    ESA \cite{huang2013esa}                     & \cellcolor{red!12} 9.16 & \cellcolor{red!5} 390 & \cellcolor{red!9} 50.1 & \cellcolor{red!20} 17.0 & \cellcolor{red!20} 5.9 & \cellcolor{red!20} 85.1 & \cellcolor{red!12} 91.0 \\
    Priority                 & \cellcolor{green!5} 11.18 & \cellcolor{red!5} 387 & \cellcolor{green!7} 62.2 & \cellcolor{red!20} 17.0 & \cellcolor{red!20} 6.0 & \cellcolor{red!20} 85.1 & \cellcolor{red!8} 91.4 \\
    \helios\ (no-ISL)        & \cellcolor{green!9} 11.85 & \cellcolor{green!4} 417 & \cellcolor{green!11} 65.7 & \cellcolor{green!4} 35.7 & \cellcolor{green!5} 58.0 & \cellcolor{green!7} 29.1 & \cellcolor{green!10} 94.4 \\
    \helios\ (no-context)    & \cellcolor{green!20} 13.38 & \cellcolor{green!20} 564 & \cellcolor{green!19} 73.1 & \cellcolor{red!6} 26.7 & \cellcolor{red!2} 43.0 & \cellcolor{green!3} 37.9 & \cellcolor{green!11} 94.7 \\
    \helios\ (noisy-context) & \cellcolor{green!20} 13.36 & \cellcolor{green!13} 508 & \cellcolor{green!20} 73.6 & \cellcolor{green!5} 36.6 & \cellcolor{green!5} 59.3 & \cellcolor{green!8} 27.9 & \cellcolor{green!20} 96.4 \\
    \helios                  & \cellcolor{green!20} 13.41 & \cellcolor{green!13} 509 & \cellcolor{green!20} 73.8 & \cellcolor{green!5} 36.7 & \cellcolor{green!5} 59.5 & \cellcolor{green!8} 27.7 & \cellcolor{green!20} 96.4 \\
    \hline
  \end{tabular}
\end{table*}

\begin{figure}[t]
\centering
\includegraphics[width=0.8\columnwidth]{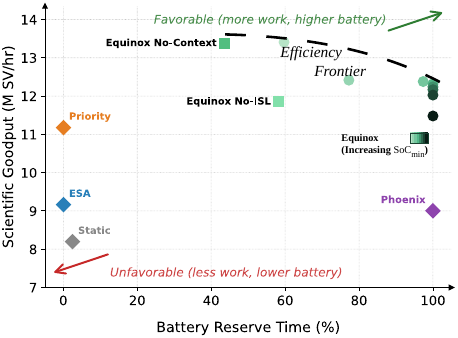}
\caption{\helios\ achieves the highest scientific goodput and event coverage of all evaluated systems while maintaining 59.5\% or more battery reserve time---the best outcome on the scientific return--operational resilience frontier.}
\label{fig:pareto}
\vspace{-3ex}
\end{figure}

Figure~\ref{fig:pareto} summarizes the central trade-off in our evaluation. The baselines occupy opposite ends of the return-resilience spectrum: Phoenix preserves battery almost perfectly but delivers relatively low scientific goodput, while Priority increases scientific return at the cost of chronic depletion. \helios\ moves this frontier outward. It is the only system that simultaneously achieves high scientific goodput and strong battery reserves, placing it on the most favorable Pareto frontier among all evaluated designs. This also showcases that tuning the parameters of \helios\ cost model (in this case, to prioritize battery more) can effectively move \helios\ along the efficiency frontier.

Table~\ref{tab:all_results} shows the full breakdown. \helios\ achieves the highest scientific goodput at 13.41~M\,SV/hr, improving on Priority by 20\% and Static by 64\%, while also delivering the highest event coverage (73.8\%). Unlike Priority, which spends only 5.9\% of timesteps with adequate battery reserve and incurs 85.1\% brownout risk, \helios\ maintains substantially stronger operational health, with 59.5\% battery reserve time and 27.7\% brownout risk. It also achieves the highest scientific load balance (96.4\%), indicating that ISL offloading distributes work cleanly across the constellation rather than concentrating execution on a small subset of satellites.

The gains in scientific goodput decompose across three architectural layers. First, \textit{value-aware selection} (Static $\rightarrow$ Priority) improves scientific goodput by 36\% by concentrating execution on higher-value scenes. Second, \textit{hardware feedback} (Priority $\rightarrow$ \helios\ (no-ISL)) substantially improves operational resilience while still increasing scientific goodput: hardware-aware scheduling prevents chronic eclipse depletion and thermal stress, increasing battery reserve time from 5.9\% to 58\%. Third, \textit{ISL offloading} (\helios\ (no-ISL) $\rightarrow$ \helios) converts locally deferred work into additional scientific return by routing tasks to less constrained neighbors, increasing scientific goodput from 11.85 to 13.41M\,SV/hr and event coverage from 65.7\% to 73.8\%.

The contextual variants clarify the role of scene priors. \textit{\helios\ (no-context)} achieves nearly the same scientific goodput as the full system, but only by processing substantially more images (564\,K/hr vs.\ 509\,K/hr) and with weaker battery health (43.0\% battery reserve time vs.\ 59.5\%, 37.9\% brownout risk vs.\ 27.7\%). Without category-conditioned context, the value gap between tasks narrows, and scheduling behaves more like resource-aware FIFO: the system still exploits hardware feedback and ISL routing, but spends more computation on lower-value images. In contrast, contextual priors allow \helios\ to preserve aggregate scientific return while processing fewer images and maintaining better battery reserves. \textit{\helios\ (noisy-context)} performs nearly identically to the full system, with 13.36\,SV/hr scientific goodput, 73.6\% coverage, and the same scientific load balance, showing that \helios\ is robust to moderate context error.

\niparagraph{Comparison with prior works.}
ESA is hardware-aware but value-blind: it matches Priority in throughput (387\,K images/hr), yet achieves 18.0\% lower scientific goodput than Priority and 31.7\% lower than \helios\ because it does not prioritize higher-value scenes. Phoenix makes the opposite trade-off, preserving battery extremely well (99.9\% battery reserve time, 0.0\% brownout risk) but yielding only 9.00\,SV/hr scientific goodput and 49.1\% coverage. Its rigid offloading structure also produces the weakest scientific load balance (90.2\%). By jointly incorporating scene value, local hardware state, and distributed offloading, \helios\ achieves substantially higher scientific return without chronic depletion.

\begin{table}
  \centering
  \caption{Per-category detection precision (\%) and total detections (K) over 72 hours, where precision reflects per-image selectivity and detections reflect operational impact.}
  \label{tab:hit_rates}
  \small
  \setlength{\tabcolsep}{3.5pt}
  \begin{tabular}{l cccc cccc}
    \toprule
    & \multicolumn{4}{c}{\textbf{Detection Precision (\%)}} & \multicolumn{4}{c}{\textbf{Events Detected (K)}} \\
    \cmidrule(lr){2-5} \cmidrule(lr){6-9}
    \textbf{System} & \footnotesize Fire & \footnotesize Flood & \footnotesize Ship & \footnotesize Monitor  & \footnotesize Fire & \footnotesize Flood & \footnotesize Ship & \footnotesize Monitor \\
    \midrule
    Static   & 6.0 & 7.2 & 5.4 & 8.5  & 1500 & 1803 & 1346 & 2126 \\
    ESA      & 6.0 & 7.2 & 5.4 & 8.5  & 1674 & 1996 & 1510 & 2393 \\
    Phoenix  & 6.0 & 7.2 & 5.4 & 8.5  & 1644 & 1961 & 1479 & 2344 \\
    Priority & 7.3 & 8.6 & 7.3 & 10.5 & 2031 & 2384 & 2037 & 2907 \\
    \helios   & 6.7 & 7.9 & 6.3 & 9.5  & \textbf{2452} & \textbf{2885} & \textbf{2315} & \textbf{3452} \\
    \bottomrule
  \end{tabular}
\end{table}

\niparagraph{Detection Quality.}
Priority improves per-image \emph{detection precision} by discarding low to mid ESV images, but quickly depletes battery reserves and shortens sustained execution through eclipse due to lack of hardware awareness. \helios\ instead trades a small reduction in per-image precision for much longer active operation, processing 31\% more images overall than Priority. As Table~\ref{tab:hit_rates} shows, this yields more absolute detections across all categories: 20.7\% more fires, 21.0\% more floods, 13.6\% more ships, and 18.7\% more environmental events. ESA and Phoenix closely track Static in detection precision, reflecting their lack of value-aware selection.

\begin{figure}
\centering
\begin{subfigure}{1\columnwidth}
    \centering
    \includegraphics[width=0.75\linewidth]{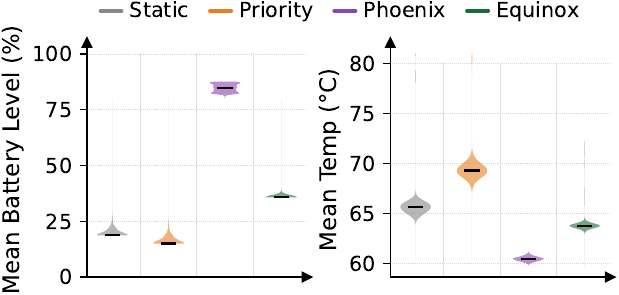}
    \caption{Fleet-wide battery and temperature distributions across systems.}
    \label{fig:soc_trajectory}
\end{subfigure}
\hfill
\hspace{2ex}
\begin{subfigure}{1\columnwidth}
    \centering
    \includegraphics[width=0.6\linewidth]{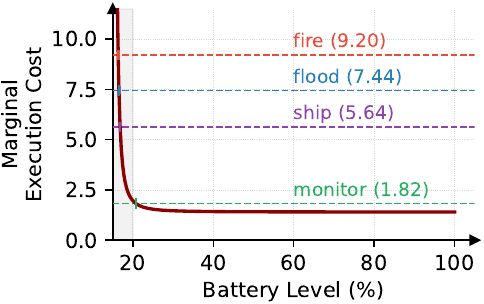}
    \caption{Battery-Aware Task Dropout.}
    \label{fig:staggered_dropout}
\end{subfigure}
\caption{Constellation-level resource behavior under sustained orbital load. (a) \helios\ maintains a stable battery regime with lower thermal stress than the baselines. (b) As battery declines, lower-value tasks are shed first, preserving reserve for higher-value work deeper into eclipse.}
\label{fig:battery_thermal_dropout}
\vspace{-2ex}
\end{figure}

\subsection{Battery Health and Thermal Management}
\label{sec:eval:battery}

Figure \ref{fig:soc_trajectory} illustrates the constellation's distributions of both battery state-of-charge (SoC) and GPU temperatures. Under the Priority baseline, battery levels cluster at the critical threshold ($\text{SoC}_c = 0.15$), confirming depletion and effectively zero operating reserve. While Phoenix preserves a high mean battery level (${>}0.86$), it demonstrates greater variance in battery distribution, reflecting uneven load accumulation across sunlit satellites. \helios\ instead strikes a balance: it maintains a healthy mean battery level with tight fleet-wide variance, indicating an efficient use of available energy rather than either exhaustion or under-utilization. As shown in Figure~\ref{fig:staggered_dropout}, this behavior stems from ordered task shedding: lower-value work drops out first as the battery declines, preserving both energy and compute capacity for higher-value tasks. Consequently, \helios\ achieves not only better battery stability but also lower thermal stress. As reflected in the Fig.~\ref{fig:soc_trajectory} temperature distributions, \helios\ reduces mean GPU temperature to 64.6$^{\circ}$C (compared to 69.5$^{\circ}$C for Priority) and caps peak temperature at 77.4$^{\circ}$C (versus 81.0$^{\circ}$C).

\subsection{Task Selection and ISL Decisions}
\label{sec:eval:isl}

\begin{table}[t]
  \centering
    \caption{Value-ordered ISL offloading by task. Higher-value tasks are offloaded under larger cost gaps, while neighbor cost remains nearly constant, indicating that offloaded work is absorbed by healthy, sunlit satellites. Offload quantity peaks at tasks with value high enough for admission to a sunlit neighbor but low-enough for rejection locally.}
  \label{tab:isl_arbitrage}
  \footnotesize
    \begin{tabular}{lcccr}
    \toprule    
    \textbf{Task} &
    \textbf{\shortstack{Avg.\\ESV}} &
    \textbf{\shortstack{Local/Neighbor\\Cost Ratio}} &
    \textbf{\shortstack{Neighbor\\Cost}} &
    \textbf{\shortstack{Offloaded\\Events (K)}} \\
    \midrule
    Monitor & 1.82 & 1.1$\times$ & 1.40 & {<}1 \\
    Ship  & 5.64 & 1.40$\times$ & 2.13 & 6{,}494 \\
    Flood & 7.44 & 1.90$\times$ & 2.15 & 152 \\
    Fire  & 9.20 & 3.53$\times$ & 2.15 & 3 \\
    \bottomrule
    \end{tabular}
\end{table}

\begin{figure}
\centering
\includegraphics[width=0.75\columnwidth]{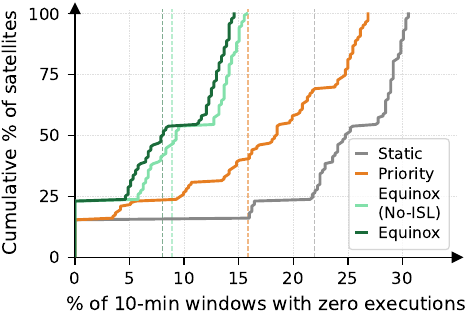}
\caption{CDF of per-satellite inactivity, measured as the fraction of 10-minute windows in which no images are processed. Dashed lines mark fleet means. \helios\ substantially reduces dark observation windows, indicating more consistent coverage across the constellation.}
\label{fig:coverage_gap}
\end{figure}

\helios\ routes $\approx$12\% of executed workload, showing that offloading is important but must be selective during the overall execution strategy. Table~\ref{tab:isl_arbitrage} confirms that ISL offloading is value-ordered. Lower-value Ship tasks (base ESV 5.64) offload under modest local stress (1.40$\times$ neighbor cost), whereas higher-value Flood and Fire tasks offload only when local execution becomes much more expensive (1.90$\times$ and 3.53$\times$). Neighbor costs stay nearly constant across all tiers, showing that offloaded work is naturally absorbed by healthy, sunlit satellites rather than redistributed among stressed nodes.

The key distinction is that \helios\ offloads \emph{selectively by value}, not merely by available capacity. Capacity-blind ISL strategies such as ESA and Phoenix can redistribute queue volume, but they do not distinguish between low-value monitoring and high-value detections. \helios\ instead uses offloading to preserve the most valuable work, improving scientific load balance from 90.2\% (Phoenix) to 96.4\% and avoiding the failure mode in which sunlit satellites spend scarce compute on low-value tasks while eclipsed neighbors are forced to drop high-value events.

Figure~\ref{fig:coverage_gap} shows the downstream effect on constellation coverage. We measure, for each satellite, the fraction of 10-minute windows in which no images are processed. Static leaves 21.9\% of windows dark on average, and Priority reduces this to 15.9\% by selecting more valuable work, but both still exhibit a broad spread across satellites. \helios\ cuts the fleet mean to 8.9\%, substantially reducing chronically inactive windows. This improvement comes from preserving battery and computation during stressed periods, which allows satellites to remain active longer and maintain more consistent geographic coverage over time.

\subsection{Demand Scaling}
\label{sec:eval:scaling}

\begin{figure}[t]
\centering
\includegraphics[width=0.80\columnwidth]{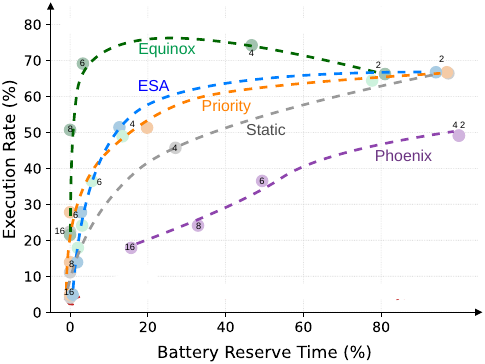}
\caption{Execution rate vs.\ battery reserve time under increasing task contention (2--16 tasks/image), with dashed trend lines. \helios\ sustains a strictly higher frontier at equal contention.}
\label{fig:task_scaling}
\Description{Scatter plot with trend lines of execution rate versus battery reserve time under increasing task contention, showing Helios maintaining a higher frontier than all baselines.}
\end{figure}

We evaluate behavior under increasing contention by scaling the number of tasks per satellite from 2 to 16, which proportionally increases both task arrival rate and per-image energy demand. As Figure~\ref{fig:task_scaling} shows, execution rate declines for every system as contention grows---this is unavoidable when demand outpaces fixed compute and energy budgets. What differentiates the systems is where their trendlines sit: \helios\ occupies a strictly higher frontier than all baselines, executing a greater fraction of arriving tasks at every contention level while maintaining comparable or better battery reserve. Only Phoenix exceeds \helios\ in battery reserve time, at the cost of far lower execution rate.

At low contention, all systems execute roughly 66\% of tasks and behave similarly. Phoenix is the outlier: it preserves battery aggressively (100\% reserve time) but executes fewer tasks (49\%), trading throughput for operational safety. As contention increases, the baselines degrade rapidly. Static and Priority drop below 15\% execution rate by 8 tasks, with battery reserve collapsing to near zero. \helios\ instead sheds lower-value work selectively before the fleet reaches depletion, sustaining 50.7\% execution at 8 tasks and 21.5\% at 16---a $5.2\times$ improvement over Static. This translates directly into scientific goodput: \helios\ reaches 39.3\,M\,SV/hr at 16 tasks, 4.4x Priority and ESA and 1.8x Phoenix. Phoenix maintains strong battery reserve throughout, but pays a steep throughput cost that widens at high contention.

\begin{figure}[t]
\centering
\includegraphics[width=0.85\columnwidth]{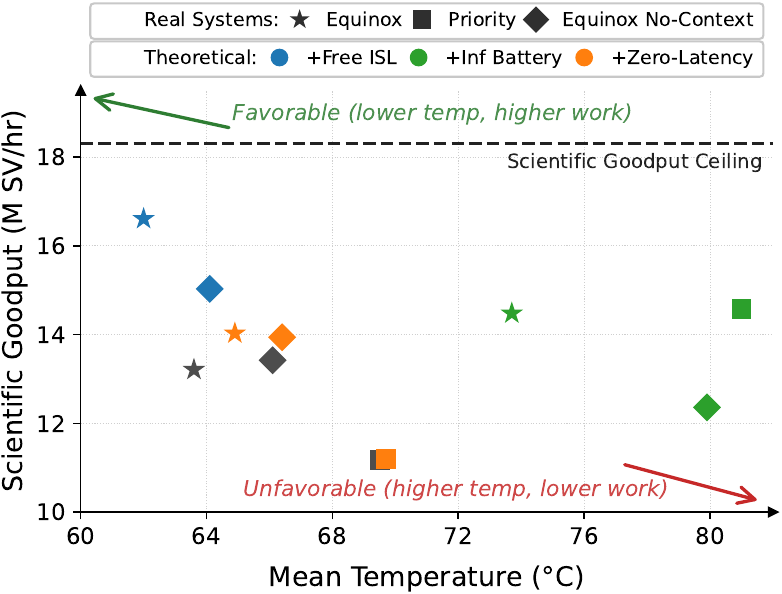}
\caption{Scientific goodput vs.\ mean temperature under individual constraint relaxations (72\,h fMoW). Each system is shown across four configurations: baseline, free ISL, infinite battery, and zero-latency GPU inference. The dashed line marks the oracle scientific goodput ceiling (18.31\,M\,SV/hr).}
\label{fig:constraint_relaxation}
\Description{Scatter plot of scientific goodput versus mean temperature for Helios, Priority, and Helios no-context under four constraint relaxations.}
\end{figure}

\subsection{Constraint Relaxation}
Figure~\ref{fig:constraint_relaxation} places each system under four hardware relaxations to isolate why \helios\ must combine value awareness with resource awareness. \textit{More execution does not automatically produce more scientific value.} Under infinite battery, \helios\ and \helios\ \textit{(no-context)} execute at nearly identical rates (67.3\% vs.\ 67.6\%), yet \helios\ delivers 17\% higher goodput (14.48 vs.\ 12.36\,M\,SV/hr) and runs 6.2$^\circ$C cooler (73.7$^\circ$C vs.\ 79.9$^\circ$C). Contextual priors continue to steer compute toward higher-value images even when energy is no longer scarce. \textit{Latency reduction helps only when energy is managed.} Zero-latency inference gives Priority essentially no goodput gain---it simply exhausts battery faster. \helios\ converts the same relaxation into a 6\% goodput increase, because it retains enough reserve to exploit the freed capacity. \textit{Network capacity amplifies value awareness but does not replace it.} With free ISL, \helios\ reaches 16.61\,M\,SV/hr at 89.9\% execution rate. Under identical routing conditions, \helios\ still outperforms \helios\ \textit{(no-context)} by 1.58\,M\,SV/hr (16.61 vs.\ 15.03): routing redistributes work, but without contextual priors more of that capacity is spent on lower-value tasks. Across all relaxations, the gap to the oracle ceiling (18.31\,M\,SV/hr) narrows most when both value awareness and resource awareness are present. A detailed per-system breakdown is in Appendix~\ref{app:relaxation}.

\begin{figure}[t]
\centering
\includegraphics[width=0.85\columnwidth]{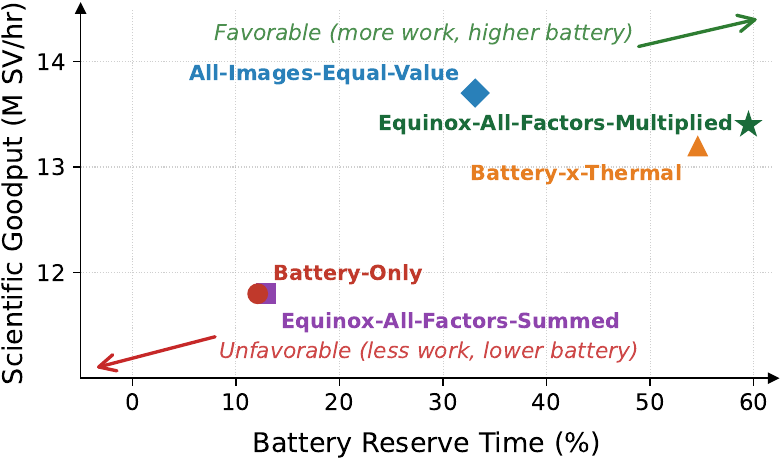}
\caption{Control signal ablation: scientific goodput vs.\ battery reserve time for each pricing variant (72\,h fMoW). The original, multiplicative \helios\ achieves the best combined outcome; additive composition and battery-only signals collapse battery reserve without improving goodput.}
\label{fig:price_ablation}
\end{figure}

\subsection{Ablation of Control Components}
Figure~\ref{fig:price_ablation} shows how each component of the \helios\ control signal contributes to the scientific goodput--battery reserve frontier.
\textit{Battery Only} is insufficient, achieving 11.8 M\,SV/hr with only 12\% battery reserve time. Adding thermal feedback raises goodput to 13.2M\,SV/hr and reserve time to 54.6\%, demonstrating that thermal moderation is not merely a safety guardrail but a prerequisite for sustained execution. The full multiplicative formulation reaches 13.4M\,SV/hr at 59.5\% reserve time; queue-depth awareness absorbs bursty arrivals before they translate into battery stress.

Signal structure matters as much as signal content. Replacing multiplication with addition (\textit{All-Factors-Summed}) reduces reserve time to 12.9\%, matching battery only, as additive composition compresses dropout thresholds that separate high- and low-value tasks, causing all work to remain active until exhaustion. 

\textit{All-Images-Equal-Value}, a scenario where all tasks and images share a uniform ESV, enables understanding the benefit of priority based offloading. This system achieves the highest raw goodput (13.7M\,SV/hr) but only 33.1\% reserve time, as without value differentiation the system cannot shed low-priority work early to preserve power. The full \helios\ design is the only variant that simultaneously achieves strong scientific return and strong operational reserve.

\begin{table}[t]
  \centering
  \caption{Jetson Orin Nano Super hardware profile. Task energy is driven primarily by latency across tiers.}
  \label{tab:hw_validation}
  \footnotesize
  \setlength{\tabcolsep}{6pt}
  \begin{tabular}{lcccc}
    \toprule
    \textbf{\shortstack{EfficientNet\\Variant}} & \textbf{\shortstack{p50\\(ms)}} & \textbf{\shortstack{p95\\(ms)}} & \textbf{\shortstack{GPU Power\\(W)}} & \textbf{\shortstack{Energy/Image\\($\mu$Wh)}} \\
    \midrule
    Light (B0)    & 32 & 33 & 3.14 & 28 \\
    Standard (B2) & 45 & 46 & 3.13 & 39 \\
    Heavy (B3)    & 51 & 52 & 3.24 & 46 \\
    \bottomrule
  \end{tabular}
\end{table}

\begin{figure}
\centering
\includegraphics[width=\columnwidth]{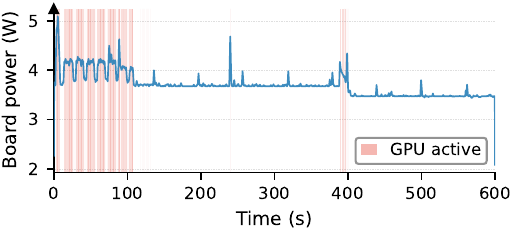}
\caption{Board power (W) during a 600-second live trace on Jetson Orin Nano Super. Red shading marks GPU-active inference bursts. Idle runtime power sits at 3.7\,W.}
\label{fig:hw_timeline}
\Description{Board power timeline over 600 seconds with red shading marking GPU-active inference bursts.}
\end{figure}

\subsection{Hardware Validation}
\label{sec:eval:hw}

Table~\ref{tab:hw_validation} grounds our simulator in measurements from a Jetson Orin Nano Super. Across task tiers, GPU power is nearly constant (3.13--3.24\,W), while energy varies primarily with execution time. Latency is also tightly concentrated within each tier, with p50--p95 spread of at most 1\,ms, supporting our use of fixed-duration task models in simulation. Under these measurements, the simulator’s 2.0\,GFLOPs/s budget maps closely to the 10\,W concurrent image-processing envelope. A 600-second live \texttt{tegrastats} trace (Figure~\ref{fig:hw_timeline}) further shows that the runtime itself imposes negligible overhead. Scheduling and admission logic run entirely on the CPU, contributing no measurable GPU utilization during task selection, while total board power remains stable at 3.5--3.8\,W and GPU temperature at 44--46$^{\circ}$C. In other words, the control logic does not consume the thermal or power budget reserved for inference, validating our assumption that scheduling overhead is negligible relative to task execution.


\subsection{Limitations}
\label{sec:eval:limitations}

Three design choices bound the current system and point to natural extensions.
\begin{itemize}[nosep,leftmargin=*]
\item \textbf{Single-hop ISL.} \helios\ forwards tasks at most once. However, since the first hop typically reaches a sunlit, low-cost neighbor, multi-hop likely offers small benefit.
\item \textbf{Fixed control parameters.} Battery, thermal, and queue coefficients are fixed at deployment. While currently robust (Figure~\ref{fig:beta_sweep_app}), dynamically adapting parameters to mission phase or workload conditions is a natural next step.
\item \textbf{Fixed task set.} Deployment-static task sets simplify scheduling but preclude joint valuation of interdependent events (e.g., correlated wildfires and flooding).
\end{itemize}
\section{Related Work}
\label{sec:related}

\niparagraph{Orbital Edge Computing and Real-Time Analytics.} OEC~\cite{orbitaledge} first demonstrated in-orbit inference on CubeSats to alleviate downlink bottlenecks. Subsequent systems refined \emph{what} to run and \emph{where}: Serval~\cite{serval} splits queries between glacial filters on the ground and dynamic filters on orbit, while EarthSight~\cite{erol2025earthsight} couples ground-station scheduling with adaptive on-orbit filter ordering. OrbitChain~\cite{li2025orbitchain}, SateLight~\cite{sate2025satelight}, and LEOEdge~\cite{yao2025leoedge} further explore workflow orchestration, model lifecycle management, and orbital offloading. Another OEC direction address compressing models to enabling more efficient inference for real-time analytics, including knowledge distillation~\cite{Gou2020KnowledgeDA}, model pruning~\cite{Zhang2022AdvancingMP, fluid}, and weight sharing~\cite{cai2020onceforall, erol2025earthsight}. Across OEC literature, however, scheduling remains mostly static, ground-directed, or decoupled from instantaneous hardware state. \helios\ addresses this gap by making execution decisions directly responsive to local battery, thermal, and queue conditions. 

\niparagraph{Energy-Aware and Safety-Critical Scheduling.}
Prior work treats energy either as a resource to optimize or as a constraint to satisfy. Kodan~\cite{kodan} and AdaEO~\cite{yang2024adaptive} reduce onboard cost through model selection and confidence-adaptive compression. Huang and Neely~\cite{huang2013esa} provide the theoretical basis for constrained energy-harvesting scheduling through ESA, while Hu et al.~\cite{hu2023lyapunov} and Zhang et al.~\cite{zhang2024peer} adapt Lyapunov methods to satellite edge systems. HyperDrive~\cite{pusztai2024hyperdrive} treats thermal state as a hard placement constraint. \helios\ instead unifies these signals into a single runtime control mechanism: local battery, thermal, and queue state jointly determine when lower-value tasks should be shed and when higher-value work should still proceed.

\niparagraph{Inter-Satellite Coordination and Mobility.}
Phoenix~\cite{liu2024phoenix} and SkyCastle~\cite{li2024skycastle} address complementary aspects of constellation coordination. Phoenix always routes tasks towards sunlit satellites to reduce battery depletion, while SkyCastle focuses on routing under fast ISL topology changes. OrbitChain~\cite{li2025orbitchain} and FedSpace~\cite{so2022fedspace} use ISLs for tracking image history and federated learning, respectively. The main limitation of prior offloading schemes is that they are driven primarily by energy state or connectivity, not by task value. \helios\ instead uses ISLs selectively to preserve high-value work under local resource stress.


\section{Conclusion}
\label{sec:conclusion}

\helios\ shows that orbital scheduling should not separate scientific prioritization from hardware state. By compressing battery, thermal, queue, and value signals into a single hardware-aware control mechanism, it enables decentralized satellites to shed low-value work early, preserve reserve through eclipse, and redirect high-value tasks to healthier neighbors. \helios\ suggests that real-time resource-aware control can serve as a practical foundation for constellation-scale scientific autonomy, with natural extensions to adaptive mission objectives, ISL routing, and downlink planning.

\balance

\bibliographystyle{ACM-Reference-Format}
\bibliography{references}

\newpage

\appendix
\section{Experimentation Environment}
\label{app:simulation}

With the mechanism and its parameters defined, this section documents the physical simulation environment and detailed per-task results to support reproducibility of Section~\ref{sec:eval}.

\subsection{Orbital Mechanics}
\label{app:orbital}

The constellation consists of 143 satellites in 13 orbital planes of 11 satellites each, at 500\,km altitude with $53^{\circ}$ inclination. Orbits are initialized from Two-Line Element (TLE) sets generated for the Walker-Delta pattern; propagation uses the SGP4 model via the \texttt{sgp4} Python library's \texttt{SatrecArray} interface, which propagates all 143 satellites in a single C-level call per timestep. Positions are returned in the True Equator Mean Equinox (TEME) reference frame; geodetic coordinates (latitude, longitude, altitude) are derived via vectorized NumPy operations. TEME differs from the Geocentric Celestial Reference System (GCRS) by at most $\sim$0.3\,arcsec, negligible for the inter-satellite distance and shadow geometry calculations used here.

\niparagraph{Eclipse detection.}
The Sun's position is computed once per timestep via the Skyfield ephemeris library. Eclipse state for all 143 satellites is then determined by a vectorized cylindrical shadow test: a satellite is in eclipse if the perpendicular distance from the satellite to the Earth--Sun line is less than Earth's radius and the satellite is on the far side of the Earth from the Sun. This avoids per-satellite Skyfield calls and reduces eclipse computation to a single matrix operation.

\niparagraph{ISL topology.}
Inter-satellite distances are computed using \texttt{scipy.spatial.distance.cdist} over TEME positions, producing the full $143 \times 143$ distance matrix in a single C-level call. Satellites within 5{,}000\,km are linked, yielding 8--12 neighbors per satellite on average. The ISL graph is represented as a NetworkX graph, recomputed each timestep.

\subsection{Orbital Heterogeneity Stress Test}
\label{app:orbital_heterogeneity}

The primary evaluation in Section~\ref{sec:eval:main} uses a homogeneous Walker-style constellation. To test whether conclusions depend on that symmetry, we evaluate three perturbations of orbital trajectories while keeping workload, hardware, and marginal execution cost calculation parameters unchanged:

\begin{itemize}
  \item \textbf{Phase Jitter:} Random mean-anomaly perturbation within each plane (non-uniform intra-plane phasing).
  \item \textbf{Altitude Tiers:} Plane-wise altitude tiers at 450/475/500/525/550\,km.
  \item \textbf{Both:} Phase jitter and altitude tiers enabled simultaneously.
\end{itemize}

Table~\ref{tab:orbital_heterogeneity} compares these variants against the Default (Even Spacing) baseline for five systems, including the \helios~(no-ISL) ablation.

\begin{table*}[ht]
  \centering
  \caption{Orbital heterogeneity stress test (72h fMoW). Default (Even Spacing) is the homogeneous Walker baseline used in Section~\ref{sec:eval:main}. Column names follow the main evaluation table where applicable: Scientific Goodput (M\,SV/hr), Throughput (K\,Images/hr), Mean Battery Level, Battery Reserve Time, and Brownout Risk.}
  \label{tab:orbital_heterogeneity}
  \footnotesize
  \setlength{\tabcolsep}{4pt}
  \begin{tabular}{l l c c c c c}
    \toprule
    \textbf{\rule{0pt}{3ex}Orbit Config\rule[-1.5ex]{0pt}{0pt}} & \textbf{\rule{0pt}{3ex}System\rule[-1.5ex]{0pt}{0pt}} & \textbf{\shortstack{\rule{0pt}{2.5ex}Scientific Goodput\\(M~SV/hr)}} & \textbf{\shortstack{\rule{0pt}{2.5ex}Throughput\\(K\,Images/hr)}} & \textbf{\shortstack{\rule{0pt}{2.5ex}Mean Battery\\Level}} & \textbf{\shortstack{\rule{0pt}{2.5ex}Battery Reserve\\Time (\%)}} & \textbf{\shortstack{\rule{0pt}{2.5ex}Brownout\\Risk (\%)}} \\
    \midrule
    \shortstack[l]{Default\\(Even Spacing)} & Static   & 8.20  & 346.0 & 0.204 & 7.4  & 64.4 \\
    & Priority & 11.18 & 387.0 & 0.170 & 5.9  & 85.1 \\
    & Phoenix  & 9.00  & 380.0 & 0.848 & 99.9 & 0.0 \\
    & \helios~(no-ISL) & 11.85 & 417.0 & 0.357 & 58.0 & 29.1 \\
    & \helios  & 13.41 & 509.0 & 0.367 & 59.5 & 27.7 \\
    \midrule
    Phase Jitter & Static   & 8.20  & 345.7 & 0.204 & 4.3  & 94.8 \\
    & Priority & 11.18 & 386.7 & 0.170 & 4.3  & 94.8 \\
    & Phoenix  & 9.04  & 381.2 & 0.848 & 99.9 & 0.0 \\
    & \helios~(no-ISL) & 11.65 & 410.0 & 0.344 & 23.6 & 33.7 \\
    & \helios  & 13.27 & 505.4 & 0.357 & 33.5 & 29.5 \\
    \midrule
    Altitude Tiers & Static   & 8.16  & 344.1 & 0.204 & 4.3  & 95.0 \\
    & Priority & 11.41 & 395.2 & 0.169 & 4.3  & 95.0 \\
    & Phoenix  & 9.09  & 383.4 & 0.850 & 99.9 & 0.0 \\
    & \helios~(no-ISL) & 11.65 & 409.9 & 0.333 & 23.6 & 33.7 \\
    & \helios  & 13.35 & 510.2 & 0.349 & 33.9 & 29.5 \\
    \midrule
    Both & Static   & 8.16  & 344.1 & 0.204 & 4.3  & 95.0 \\
    & Priority & 11.42 & 395.6 & 0.169 & 4.3  & 95.0 \\
    & Phoenix  & 9.04  & 381.4 & 0.851 & 99.9 & 0.0 \\
    & \helios~(no-ISL) & 11.65 & 409.9 & 0.333 & 23.6 & 33.7 \\
    & \helios  & 14.56 & 584.6 & 0.360 & 38.7 & 28.5 \\
    \bottomrule
  \end{tabular}
\end{table*}

\niparagraph{Interpretation.}
The results are robust to orbital heterogeneity. Across all four orbit configurations, \helios\ remains the highest-scientific-goodput system, and its advantage over the baselines persists under both phase and altitude perturbations. The \helios\ (no-ISL) ablation changes little across the perturbed cases, whereas full \helios\ continues to benefit from heterogeneity, showing that ISL coordination is the mechanism that converts orbital diversity into additional scientific return. Phoenix preserves battery almost perfectly in every setting, but again does so at the cost of substantially lower scientific goodput, reproducing the same return--resilience trade-off as in the main evaluation. Notably, the combined perturbation case (\textit{Both}) improves \helios\ most strongly, raising scientific goodput to 14.56\,SV/hr and throughput to 584.6\,K images/hr, suggesting that greater orbital asymmetry creates more opportunities for value-aware offloading.

\subsection{Hardware Models}
\label{app:hardware}

\niparagraph{Battery model.}
Each satellite carries a 100\,Wh lithium-ion battery modeled as an energy reservoir with instantaneous charge/discharge:
\[
  \text{SoC}_{t+1} = \text{SoC}_t + \frac{(P_{\text{solar}} - P_{\text{load}}) \cdot \Delta t}{C_{\text{batt}}}
\]
where $P_{\text{solar}} = \min(P_{\text{max}} \cdot \cos\theta_{\text{sun}},\; P_{\text{max}})$ when sunlit and 0 during eclipse, $P_{\text{load}} = P_{\text{idle}} + N_{\text{tasks}} \cdot P_{\text{task}}$ is total power draw, and $C_{\text{batt}} = 100$\,Wh. SoC is clamped to $[0, 1]$.

\niparagraph{Thermal model.}
Die temperature follows a first-order RC thermal model with radiative cooling:
\[
  T_{t+1} = T_t + \frac{\Delta t}{\tau} \cdot (T_{\text{eq}} - T_t)
\]
where $\tau$ is the thermal time constant and $T_{\text{eq}}$ is the equilibrium temperature at current power dissipation. In vacuum, heat rejection is by radiation only ($\propto T^4$), and the equilibrium temperature can \emph{decrease} below current temperature when power dissipation drops---this is correct RC model behavior, not a bug.

\begin{table}[ht]
  \centering
  \caption{Representative fMoW categories and event probabilities. Full mapping covers 62 categories; default probabilities apply to unlisted categories.}
  \label{tab:fmow_categories}
  \small
  \begin{tabular}{lrrrr}
    \toprule
    \textbf{Category} & \textbf{$P_{\text{fire}}$} & \textbf{$P_{\text{flood}}$} & \textbf{$P_{\text{ship}}$} & \textbf{$P_{\text{monitor}}$} \\
    \midrule
    flooded\_road       & 0.01 & \textbf{0.90} & 0.01 & 0.15 \\
    shipyard            & 0.03 & 0.03 & \textbf{0.80} & 0.05 \\
    port                & 0.03 & 0.05 & \textbf{0.70} & 0.08 \\
    smokestack          & \textbf{0.35} & 0.03 & 0.02 & 0.12 \\
    dam                 & 0.02 & \textbf{0.40} & 0.05 & 0.10 \\
    military\_facility  & 0.08 & 0.05 & 0.08 & \textbf{0.40} \\
    nuclear\_powerplant & 0.10 & 0.05 & 0.02 & \textbf{0.35} \\
    oil\_or\_gas\_fac.  & \textbf{0.30} & 0.03 & 0.05 & 0.15 \\
    lighthouse          & 0.02 & 0.10 & \textbf{0.35} & 0.06 \\
    factory\_or\_pp.    & \textbf{0.25} & 0.04 & 0.02 & 0.15 \\
    lake\_or\_pond      & 0.02 & \textbf{0.25} & 0.08 & 0.06 \\
    space\_facility     & 0.08 & 0.03 & 0.02 & \textbf{0.30} \\
    prison              & 0.06 & 0.04 & 0.02 & \textbf{0.25} \\
    crop\_field         & 0.06 & 0.10 & 0.02 & 0.06 \\
    \midrule
    \textit{\_default}  & \textit{0.05} & \textit{0.08} & \textit{0.12} & \textit{0.10} \\
    \bottomrule
  \end{tabular}
\end{table}

\subsection{fMoW Category Mapping}
\label{app:fmow}

Table~\ref{tab:fmow_categories} lists representative categories from the functional Map of the World (fMoW) dataset and their task-conditioned event probabilities. The full mapping covers 62 categories; we show the 15 with the greatest variation across tasks to highlight the scene heterogeneity that \helios\ exploits through category-conditioned context.
This heterogeneity is substantial: for example, $P(\text{ship} \mid \text{shipyard}) = 0.80$ whereas $P(\text{ship} \mid \text{crop\_field}) = 0.02$, a 40$\times$ difference. Such variation is what makes contextual ranking effective. A vessel-detection task applied to a shipyard image is far more likely to produce useful scientific output than the same task applied to an agricultural scene, so category-conditioned estimates help \helios\ focus limited onboard compute and energy on the most promising opportunities. Without these priors (the no-context ablation), images within the same task class become artificially indistinguishable, and scheduling loses an important source of selectivity, behaving much more like resource-limited FIFO within each task type.


\section{Parameter Calibration}
\label{app:calibration}

The scheduling mechanism (Section~\ref{sec:system}) uses several fixed parameters. This section justifies each choice and characterizes sensitivity where applicable. All parameters are set at deployment; Appendix~\ref{app:sensitivity} evaluates robustness to the most important one ($\beta$).

\subsection{Barrier Coefficient \texorpdfstring{$\beta$}{Beta}}
\label{app:beta}

The barrier coefficient $\beta = 0.001$ controls how sharply the marginal execution cost rises as SoC approaches the critical threshold SoC$_c = 0.15$. Using Eq.~\ref{eq:dropout_soc}, Table~\ref{tab:dropout_soc} gives the closed-form dropout SoC for each task under the default parameters ($P_{\text{base}} = 1.40$, SoC$_c = 0.15$).

\begin{table}[ht]
  \centering
  \caption{Analytic dropout thresholds at $\beta = 0.001$, $P_\text{base} = 1.40$, $\text{SoC}_c = 0.15$.}
  \label{tab:dropout_soc}
  \small
  \begin{tabular}{lrr}
    \toprule
    \textbf{task} & \textbf{ESV} & \textbf{SoC$^*$} \\
    \midrule
    Env.\ monitoring & 1.82 & 0.208 \\
    Vessel detection & 5.64 & 0.168 \\
    Flood detection  & 7.44 & 0.165 \\
    Fire detection   & 9.20 & 0.163 \\
    \bottomrule
  \end{tabular}
\end{table}

The separation between the monitor task (SoC$^* = 0.208$) and the three higher-value tasks (SoC$^* \in [0.163, 0.168]$) spans approximately 4 percentage points of SoC---roughly 4\,Wh on a 100\,Wh battery, or $\sim$2 minutes of eclipse-phase drain at 120\,W load. This is large enough for the monitor task to be consistently excluded before the others, but small enough that the transition occupies a narrow SoC band rather than wasting capacity. At $\beta = 0.01$ (10$\times$ larger), the monitor dropout shifts to SoC$^* = 0.33$, sacrificing significant sunlit throughput. At $\beta = 0.0001$ (10$\times$ smaller), the monitor dropout falls to SoC$^* = 0.168$, barely above the higher-value tasks, and the staggered dropout effect is compressed into a $<$2\,Wh window.

Figure~\ref{fig:beta_sweep_app} in Appendix~\ref{app:sensitivity} shows that despite this sensitivity in \emph{where} dropout occurs, Scientific Goodput varies by only $\pm$10\% across a 100$\times$ range of $\beta$. The mechanism is self-correcting: increasing $\beta$ raises execution costs earlier, triggering self-curtailment that increases SoC headroom and moderates $f_{\text{batt}}$.

\subsection{Queue Backpressure \texorpdfstring{$\gamma$}{Gamma}}
\label{app:gamma}

The queue factor $f_{\text{queue}}(q) = 1 + \gamma_q q$ with $\gamma_q = 0.01$ is deliberately linear. A quadratic or exponential form would create a sharp cost jump when the first deferred image arrives, suppressing task execution and causing the queue to drain, which lowers the execution cost, re-admitting tasks, refilling the queue---an oscillation cycle. The linear form provides proportional backpressure: at $q = 10$ images, the cost markup is 10\%; at $q = 100$, it doubles. This keeps queue growth stable while applying meaningful cost when backlogs accumulate during extended eclipse periods.

The constant $\gamma = 0.01$ is calibrated so that the deferred queue depth at which the queue factor alone would exclude the monitor task ($\text{ESV} = 1.82$) is approximately $q = 30$ images:
\[
  P_{\text{base}} \cdot (1 + \gamma q) = 1.82 \implies q = \frac{1.82/1.40 - 1}{0.01} = 30
\]
This ensures that moderate backlogs (10--20 images) create cost pressure without immediately excluding any task, while sustained backlogs ($>$30) begin to trigger value-ordered shedding.

\subsection{ISL Cost Parameters}
\label{app:isl_cost}

The offloading condition (Section \ref{sec:system:isl}) uses two cost terms:

\niparagraph{Link cost ($c_{\text{link}}$).}
$c_{\text{link}} = 0.05 \cdot P_\ell$, proportional to the local marginal execution cost. At the base execution cost (1.40), this is 0.07 per hop. At eclipse-stressed execution costs (5--10), it rises to 0.25--0.50. The proportional form ensures that ISL routing requires a \emph{relative} cost advantage, not just an absolute one: a 5\% cost differential is the minimum for routing to be profitable at any cost level. This prevents spurious routing during sunlit periods when both local and neighbor execution costs are near baseline.

\niparagraph{Latency penalty ($c_{\text{lat}}$).}
$c_{\text{lat}} = 0.10$ flat per hop. At 5{,}000\,km maximum ISL range and optical link speeds, one-way propagation delay is ${\sim}$17\,ms---negligible for image processing but non-zero for time-critical detections. The flat 0.10 penalty represents the opportunity cost of delayed processing. Combined with $c_{\text{link}}$, the total ISL cost at the base execution cost is 0.17 per hop, requiring $P_\ell \geq P_n + 0.17$ for routing to trigger---a condition met only during eclipse-induced battery stress.

\niparagraph{Multi-hop cost accumulation.}
A second ISL hop would add another $0.05 \cdot P_{n_1} + 0.10 \approx 0.17$ at baseline-cost neighbors. Since the first-hop destination is typically a sunlit satellite with execution cost near $P_\text{base}$, the cost differential between the first-hop destination and its neighbors is small ($<$0.10), making a second hop unprofitable except under extreme, sustained eclipse at the highest-ESV tiers. This justifies the single-hop design (Section~\ref{sec:eval:limitations}).

\subsection{Deferred Queue TTL}
\label{app:ttl}

The TTL of 300 seconds ($\approx$5 minutes) is calibrated to the eclipse-to-sunlit transition. A satellite at 500\,km LEO has an orbital period of $\sim$94 minutes with $\sim$33 minutes in eclipse. The transition from peak eclipse stress to sunlit recovery takes approximately 10--15 minutes as the satellite crosses the terminator and solar panels begin recharging. A 300-second TTL spans roughly one-third of this transition, giving deferred images a reasonable window to encounter lower execution costs as the battery recharges, without retaining stale images across a full orbital cycle.

If the TTL were much shorter ($<$60\,s), deferred images would expire before the satellite reaches sunlight, negating the temporal demand-shifting benefit. If much longer ($>$900\,s), images from the previous eclipse cycle would compete with fresh arrivals in the next sunlit phase, creating artificial congestion.

\subsection{Safe Reserve Threshold \texorpdfstring{SoC$_{\text{safe}}$}{SoC\_safe}}
\label{app:soc_safe}

The battery reserve time metric (Section \ref{sec:system:metrics} uses SoC$_{\text{safe}} = 0.35$ as the operator-defined reserve threshold. This value represents the minimum charge at which a satellite can handle an unplanned high-priority tasking burst---for example, a wildfire detection request that arrives during eclipse. At 100\,Wh battery capacity, SoC $= 0.35$ provides 40\,Wh of reserve, sufficient for approximately 20 minutes of continuous processing at 120\,W task load (accounting for concurrent solar charging on lit satellites). The choice is conservative: a satellite with SoC $> 0.35$ can execute any arriving task without risk of breaching SoC$_c = 0.15$, even if the task arrives during the worst-case eclipse phase.

\section{Constraint Relaxation Ablation}
\label{app:relaxation}

Table~\ref{tab:relaxation_full} reports all eight metrics for three systems---
\helios, \helios~(no-prior), and Priority---under four independent constraint
relaxations over a 72-hour fMoW simulation. Each row removes exactly one resource
limit; Oracle removes all simultaneously. The \textit{no-prior} variant uses
uniform base-rate task allocation with no category-conditioned ESV; it is otherwise
identical to \helios.

\begin{table*}[ht]
  \centering
  \caption{Full constraint relaxation results (72\,h fMoW, AGX Orin). Science Load Balance is higher when more uniform across satellites.}
  \label{tab:relaxation_full}
  \small
  \begin{tabular}{@{}lrrrrrrr@{}}
    \toprule
    \textbf{\rule{0pt}{3ex}System\rule[-1.5ex]{0pt}{0pt}} & \textbf{\shortstack{\rule{0pt}{2.5ex}Scientific Goodput\\(M~SV/hr)}}
      & \textbf{\shortstack{\rule{0pt}{2.5ex}Event\\Coverage (\%)}}
      & \textbf{\shortstack{\rule{0pt}{2.5ex}Executed\\Images (M)}}
      & \textbf{\shortstack{\rule{0pt}{2.5ex}Execution\\Rate (\%)}}
      & \textbf{\shortstack{\rule{0pt}{2.5ex}Mean Battery\\Level}}
      & \textbf{\shortstack{\rule{0pt}{2.5ex}Mean Temp\\(°C)}}
      & \textbf{\shortstack{\rule{0pt}{2.5ex}Science Load\\Balance (\%)}} \\
    \midrule
    \helios                    & 13.21 & 72.6 & 36.1 & 65.0 & 0.356 & 64.6 & 96.1 \\
    \quad+Free ISL             & 16.61 & 90.7 & 50.0 & 89.9 & 0.409 & 62.0 & 96.2 \\
    \quad+Inf.\ Battery        & 14.48 & 80.0 & 37.4 & 67.3 & 1.000 & 73.7 & 99.7 \\
    \quad+Instant CPU          & 14.03 & 76.5 & 42.0 & 75.6 & 0.181 & 66.3 & 93.3 \\
    \quad Oracle               & 18.31 & 100  & 55.6 & 100  & 0.968 & 35.0 & 99.8 \\
    \midrule
    \helios~(no-prior)         & 13.42 & 73.3 & 40.7 & 73.3 & 0.264 & 66.1 & 94.7 \\
    \quad+Free ISL             & 15.03 & 82.1 & 45.6 & 82.1 & 0.358 & 64.1 & 94.2 \\
    \quad+Inf.\ Battery        & 12.36 & 67.5 & 37.5 & 67.6 & 1.000 & 79.9 & 99.7 \\
    \quad+Instant CPU          & 13.94 & 76.1 & 42.3 & 76.2 & 0.172 & 66.4 & 93.1 \\
    \quad Oracle               & 18.31 & 100  & 55.6 & 100  & 0.968 & 35.0 & 99.8 \\
    \midrule
    Priority                   & 11.18 & 61.9 & 27.8 & 50.1 & 0.170 & 69.5 & 91.4 \\
    \quad+Free ISL             & ---   & ---  & ---  & ---  & ---   & ---  & ---  \\
    \quad+Inf.\ Battery        & 14.57 & 80.5 & 37.1 & 66.7 & 1.000 & 81.0 & 99.8 \\
    \quad+Instant CPU          & 11.19 & 61.9 & 27.9 & 50.3 & 0.154 & 69.7 & 91.3 \\
    \quad Oracle               & 18.31 & 100  & 55.6 & 100  & 0.968 & 35.0 & 99.8 \\
    \bottomrule
  \end{tabular}
\end{table*}

\niparagraph{Joint hardware and value awareness jointly determine Scientific Goodput.}
The infinite battery rows expose the fundamental gap between the systems.
\helios~and no-prior execute at nearly identical rates (67.3\% vs.\ 67.6\%), yet
\helios~delivers 17\% higher Scientific Goodput (14.48 vs.\ 12.36\,M\,SV/hr). With equal access to energy,
the differentiator is which images are chosen: \helios's category-conditioned ESV
concentrates compute on scenes with high estimated event probability, while no-prior
applies uniform values that treats a port and a crop field identically.
At baseline, both systems are energy-constrained and the Scientific Goodput gap is narrow (13.21 vs.\
13.42\,M\,SV/hr), as no-prior's higher execution rate (73.3\% vs.\ 65.0\%) partially
compensates for its indiscriminate selection.
Once energy is unlimited, execution rate equalization isolates selection quality as the
sole differentiator, and the 17\% Scientific Goodput gap emerges.

\niparagraph{Thermal management is a consequence of value-aware scheduling.}
At equal execution volume under infinite battery, \helios~runs 6.2\,$^{\circ}$C
cooler than no-prior (73.7 vs.\ 79.9\,$^{\circ}$C) and 7.3\,$^{\circ}$C cooler
than Priority (73.7 vs.\ 81.0\,$^{\circ}$C). This is not a separate thermal policy:
\helios's joint scheduling signal ($f_{\text{therm}}$ coupled with $f_{\text{batt}}$
and per-task ESV) naturally gates execution when temperature rises, pacing compute
load within the safe thermal band described in Section~\ref{sec:eval:battery}.
No-prior and Priority, lacking this signal, run unconstrained at full load whenever
energy permits, resulting in higher temperature and---for no-prior---active
suppression of future tasks as rising temperature pushes $f_{\text{therm}}$ upward,
contributing to its Scientific Goodput regression relative to baseline (12.36 vs.\ 13.42\,M\,SV/hr) under
infinite battery. This thermal gap vanishes under the oracle (all systems reach
35.0\,$^{\circ}$C), confirming it is a product of scheduling policy under load,
not hardware asymmetry.

\niparagraph{Energy management is Priority's binding policy gap, not compute.}
Providing Priority with instant CPU yields zero Scientific Goodput improvement (805\,M in both
cases); mean SoC decreases from 0.170 to 0.154, as faster processing accelerates
energy drain with no mechanism to moderate pace. This directly confirms the
diagnosis from Section~\ref{sec:eval:battery}: Priority depletes to SoC$_c$ before
exhausting its compute budget, leaving available cycles idle at the energy floor.
The contrast with \helios's instant CPU result is informative: \helios~gains 6\%
Scientific Goodput (13.21$\to$14.03\,M\,SV/hr) and 10.6\,pp in execution rate (65.0$\to$75.6\%), because its
battery barrier preserves operating headroom that additional compute can actually
fill. Compute relaxation is productive only when energy management prevents the
system from reaching the floor in the first place.

\niparagraph{Free ISL enables load balancing, but amplifies rather than substitutes
for value awareness.}
Unrestricted ISL connectivity allows satellites to route images to any neighbor at
no cost, enabling optimal geographic load distribution across the fleet. For
\helios, this raises execution rate from 65.0\% to 89.9\% and event coverage from 72.6\% to
90.7\%. The no-prior variant sees smaller gains (73.3$\to$82.1\% execution rate, 15.03 vs.\
16.61\,M\,SV/hr Scientific Goodput), despite identical ISL conditions: the Scientific Goodput gap reflects \helios's
ability to route selectively by task value, concentrating high-probability images on
the satellites best positioned to process them. Priority has no ISL capability and therefore has no +Free ISL entry.
Science Load Balance approaches 99.8\% under all free-ISL
and oracle configurations, confirming that per-satellite load imbalance in baseline
systems arises from geographic eclipse asymmetry, which free ISL eliminates. In the
operational setting where ISL bandwidth carries real cost, value-aware routing
determines how much of this balancing potential is realized without sacrificing
selection quality.


\section{Sensitivity Analysis}
\label{app:sensitivity}

This section complements the main ablation studies with sensitivity analysis complementing system design exploration.

\subsection{Barrier Coefficient \texorpdfstring{$\beta$ Robustness}{Beta Robustness}}
\label{app:sensitivity:beta}

The barrier coefficient $\beta$ dictates cost sensitivity to battery depletion. We sweep $\beta$ across a $100\times$ range ($10^{-4}$ to $10^{-2}$). As Figure~\ref{fig:beta_sweep_app} demonstrates, Scientific Goodput varies by only $\pm$10\% across the full range. The mechanism is self-correcting: an artificially high $\beta$ raises execution costs earlier, forcing tasks out, which safely increases SoC headroom (mean fleet SoC rises from 0.26 to 0.34). This elevated SoC lowers $f_\text{batt}$ and moderates further cost escalation. The system achieves stable performance near the default ($\beta = 10^{-3}$) without brittle parameter tuning.

\begin{figure}[t]
\centering
\includegraphics[width=0.85\columnwidth]{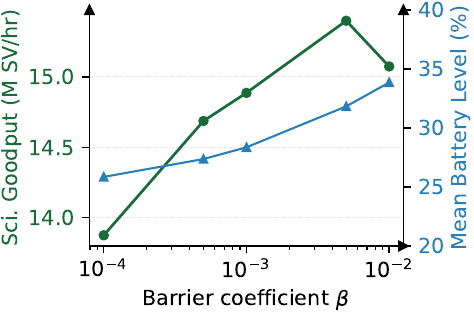}
\caption{Sensitivity to barrier coefficient $\beta$ (24h, 4 tasks). Scientific Goodput varies by $\pm$10\% across a 100$\times$ range. Higher $\beta$ safely raises mean battery level (right axis) due to a self-correcting cost feedback loop.}
\label{fig:beta_sweep_app}
\Description{Line plot of Scientific Goodput and mean battery level versus barrier coefficient beta on a log scale.}
\end{figure}

\subsection{ISL Link Failure Robustness}
\label{app:sensitivity:isl_failure}

To test graceful degradation to ISL failures, we inject independent link failures uniformly across the constellation. Figure~\ref{fig:isl_failure_app} shows Scientific Goodput remains stable ($<$0.1\% variation) up to 50\% link failure. Because ISL offloading handles only $\approx$12\% of total traffic and dense LEO topology provides 8--12 original neighbors per satellite, 50\% connectivity loss leaves sufficient path redundancy for nominal operations. Performance degrades significantly ($-34\%$ Scientific Goodput) only under complete ISL loss ($100\%$), proving the system relies on offloading for optimization but not baseline survival.

\begin{figure}[t]
\centering
\includegraphics[width=0.85\columnwidth]{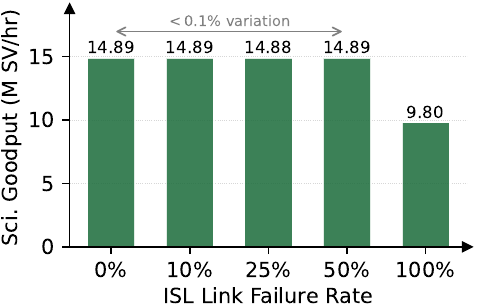}
\caption{Robustness to ISL link failures. Scientific Goodput remains stable through 50\% link failure. Significant degradation ($-34\%$ Scientific Goodput) occurs only at complete ISL loss.}
\label{fig:isl_failure_app}
\Description{Bar chart of Scientific Goodput versus ISL failure rate from 0 percent to 100 percent.}
\end{figure}

\section{Formal Properties}
\label{app:formal}

Section~\ref{sec:system:dropout} derives the critical SoC at which each task is excluded from execution (Eq.~\ref{eq:dropout_soc}) and claims that this threshold is monotonically decreasing in ESV, producing value-ordered dropout. We prove this claim below, then establish that the greedy clearing algorithm is exactly optimal under the uniform-compute assumption.

\subsection{Staggered Dropout: Proof of Value-Ordered Exit}
\label{app:dropout_proof}

\begin{theorem}[Value-ordered dropout]
\label{thm:dropout}
Let tasks $a$ and $b$ have expected science values $\mathrm{ESV}_a < \mathrm{ESV}_b$, both strictly greater than the base execution cost $P_{\mathrm{base}}$. Under the barrier cost function $f_{\mathrm{batt}}(\mathrm{SoC}) = 1 + \beta / (\mathrm{SoC} - \mathrm{SoC}_c)^2$ with $\beta > 0$, the critical SoC values satisfy $\mathrm{SoC}_a^* > \mathrm{SoC}_b^*$. That is, the lower-value task is excluded at a strictly higher state of charge.
\end{theorem}

\begin{proof}
From Eq.~\ref{eq:dropout_soc}, the critical SoC for task $a$ is:
\[
  \mathrm{SoC}_a^* = \mathrm{SoC}_c + \sqrt{\frac{\beta \cdot P_{\mathrm{base}}}{\mathrm{ESV}_a - P_{\mathrm{base}}}}
\]

Define $h(v) = \mathrm{SoC}_c + \sqrt{\beta P_{\mathrm{base}} / (v - P_{\mathrm{base}})}$ for $v > P_{\mathrm{base}}$. We show $h$ is strictly decreasing. The derivative is:
\[
  h'(v) = -\frac{1}{2} \cdot \frac{\sqrt{\beta P_{\mathrm{base}}}}{(v - P_{\mathrm{base}})^{3/2}} < 0
\]

since $\beta, P_{\mathrm{base}} > 0$ and $v > P_{\mathrm{base}}$. Therefore $\mathrm{ESV}_a < \mathrm{ESV}_b$ implies $h(\mathrm{ESV}_a) > h(\mathrm{ESV}_b)$, i.e., $\mathrm{SoC}_a^* > \mathrm{SoC}_b^*$.
\end{proof}

\begin{corollary}[Uniqueness of crossing point]
\label{cor:unique_crossing}
Each task $a$ has exactly one SoC value at which $P = \mathrm{ESV}_a$. For $\mathrm{SoC} > \mathrm{SoC}_a^*$ the task is admitted; for $\mathrm{SoC} < \mathrm{SoC}_a^*$ it abstains. There is no oscillation.
\end{corollary}

\begin{proof}
The marginal execution cost $P(\mathrm{SoC}) = P_{\mathrm{base}} \cdot [1 + \beta / (\mathrm{SoC} - \mathrm{SoC}_c)^2]$ is strictly decreasing on $(\mathrm{SoC}_c, \infty)$, since $f_{\mathrm{batt}}$ is strictly decreasing. A strictly monotone function crosses any horizontal line $\mathrm{ESV}_a$ at most once. Existence follows from the intermediate value theorem: $P \to \infty$ as $\mathrm{SoC} \to \mathrm{SoC}_c^+$ and $P \to P_{\mathrm{base}}$ as $\mathrm{SoC} \to \infty$, so for any $\mathrm{ESV}_a > P_{\mathrm{base}}$ there is exactly one crossing.
\end{proof}

\begin{corollary}[Queue-shifted invariance]
\label{cor:queue_shift}
When the deferred queue is non-empty ($q > 0$), the effective marginal execution cost becomes $P' = P \cdot (1 + \gamma q)$. The critical SoC shifts upward uniformly for all tasks:
\[
  \mathrm{SoC}_a^{*\prime} = \mathrm{SoC}_c + \sqrt{\frac{\beta \cdot P_{\mathrm{base}} \cdot (1 + \gamma q)}{\mathrm{ESV}_a - P_{\mathrm{base}} \cdot (1 + \gamma q)}}
\]
The monotonicity argument is identical (the numerator scales by a positive constant and the denominator remains positive for all tasks still available), so the dropout \emph{order} is preserved. Queue congestion makes all tasks exit earlier but does not alter which task exits first.
\end{corollary}

\subsection{Optimality of Greedy Task Selection}
\label{app:knapsack}

Section~\ref{sec:system:bidding} uses a greedy-by-value-density heuristic to select tasks. With $k$ constant compute tiers, the knapsack has at most $k$ distinct item weights and admits an exact DP solution in $O(nB)$ time. We show below that the greedy heuristic is itself exactly optimal in the common case.

\begin{proposition}[Exact optimality under uniform compute]
\label{prop:knapsack}
When all tasks consume the same compute $g_i = g$ for all $i$, the greedy-by-ESV algorithm solves the single-constraint knapsack \emph{exactly}.
\end{proposition}

\begin{proof}
With uniform item size $g$, the budget $G$ admits at most $m = \lfloor G/g \rfloor$ items. Selecting the $m$ highest-ESV tasks is optimal by a simple exchange argument: swapping any selected task for a non-selected one with lower ESV can only decrease total value.
\end{proof}

\niparagraph{Dual-constraint case.}
In practice, task selection enforces both a compute constraint ($\sum g_i \leq G$) and an energy constraint ($\sum e_i \leq E$). When both $g_i$ and $e_i$ are uniform across tasks, the dual constraint reduces to $|S| \leq \min(\lfloor G/g \rfloor, \lfloor E/e \rfloor)$, and greedy-by-ESV remains exactly optimal. When tasks select heterogeneous tiers, the problem has at most $k$ distinct weight classes; an exact solution is available via DP in $O(nB)$ time, but we observe that greedy-by-ESV-density matches DP in all evaluated configurations.

\section{Algorithms}
\label{app:algorithms}

Appendix~\ref{app:formal} establishes the theoretical properties of staggered dropout and greedy task selection. The pseudocode below shows how the runtime realizes them. All four procedures execute sequentially within each satellite's scheduling step; the full loop is given in Algorithm~\ref{alg:scheduling_step}.

\begin{algorithm}[t]
\caption{Per-Satellite Scheduling Step}
\label{alg:scheduling_step}
\DontPrintSemicolon
\KwIn{Hardware state (SoC, $T$), arriving images $\mathcal{I}$, deferred queue $\mathcal{D}$, ISL neighbors $\mathcal{N}$}
\KwOut{Executed images, updated $\mathcal{D}$}

\tcp{Step 1: Compute marginal execution cost (Eqs.~\ref{eq:spot_price}--\ref{eq:f_queue})}
$P \gets P_{\text{base}} \cdot f_{\text{batt}}(\text{SoC}) \cdot f_{\text{thermal}}(T) \cdot f_{\text{queue}}(|\mathcal{D}|)$\;

\tcp{Step 2: Collect and Assign Tasks}
$\mathcal{B} \gets \emptyset$\;
\ForEach{image $i \in \mathcal{I} \cup \textsc{RetryDeferred}(\mathcal{D}, P)$}{
  \ForEach{task $a$}{
    \If{$\mathrm{ESV}_a(i) > P$}{
      $\mathcal{B} \gets \mathcal{B} \cup \{(a, i, \mathrm{ESV}_a(i))\}$\;
    }
  }
}
$\mathcal{W}, \mathcal{L} \gets \textsc{GreedyKnapsack}(\mathcal{B}, G, E)$ \tcp*{Alg.~\ref{alg:clearing}}

\tcp{Step 3: Execute winners}
\ForEach{$(a, i, w) \in \mathcal{W}$}{
  Execute image $i$ for task $a$; update SoC, $T$\;
}

\tcp{Step 4: Manage deferred queue}
\ForEach{image $i$ without winning tasks}{
  \If{$i$ is fresh}{
    $\mathcal{D} \gets \mathcal{D} \cup \{(i, \text{TTL}=300\text{s})\}$\;
  }
}
Expire entries in $\mathcal{D}$ where TTL $\leq 0$\;

\tcp{Step 5: ISL Offloading}
$P^{\text{phys}}_\ell \gets P_{\text{base}} \cdot f_{\text{batt}}(\text{SoC}) \cdot f_{\text{thermal}}(T) \cdot f_{\text{queue}}(0)$\;
\ForEach{fresh-event loser $(a, i, w) \in \mathcal{L}$}{
  \textsc{ISLOffloading}$(i, P^{\text{phys}}_\ell, \mathcal{N})$ \tcp*{Alg.~\ref{alg:isl}}
}
\end{algorithm}

\begin{algorithm}[t]
\caption{Greedy Knapsack Task Selection}
\label{alg:clearing}
\DontPrintSemicolon
\KwIn{Bids $\mathcal{B} = \{(a_i, \text{img}_i, \text{ESV}_i)\}$, compute budget $G$, energy budget $E$}
\KwOut{Winners $\mathcal{W}$, losers $\mathcal{L}$}

Sort $\mathcal{B}$ by $\text{ESV}_i$ descending \tcp*{uniform $g_i$ $\Rightarrow$ rank by ESV}
$\mathcal{W} \gets \emptyset$;\ \ $G_{\text{rem}} \gets G$;\ \ $E_{\text{rem}} \gets E$\;
\ForEach{task-bid $(a_i, \text{img}_i, \text{ESV}_i)$ in sorted order}{
  \If{$g \leq G_{\mathrm{rem}}$ \textbf{and} $e \leq E_{\mathrm{rem}}$}{
    $\mathcal{W} \gets \mathcal{W} \cup \{(a_i, \text{img}_i, \text{ESV}_i)\}$\;
    $G_{\text{rem}} \gets G_{\text{rem}} - g$;\ \ $E_{\text{rem}} \gets E_{\text{rem}} - e$\;
  }
}
$\mathcal{L} \gets \mathcal{B} \setminus \mathcal{W}$\;
\Return $\mathcal{W}, \mathcal{L}$\;
\end{algorithm}

\begin{algorithm}[t]
\caption{ISL Offloading}
\label{alg:isl}
\DontPrintSemicolon
\KwIn{Image $i$, local execution cost $P^{\text{phys}}_\ell$, ISL neighbors $\mathcal{N}$}

$n^* \gets \textsc{None}$;\ \ $P^*_{\text{adj}} \gets P^{\text{phys}}_\ell$\;
\ForEach{neighbor $n \in \mathcal{N}$}{
  $P^{\text{phys}}_n \gets P_{\text{base}} \cdot f_{\text{batt,n}} \cdot f_{\text{thermal,n}} \cdot f_{\text{queue}}(0)$\;
  $P_{\text{adj}} \gets P^{\text{phys}}_n + c_{\text{link}}(d_n) + c_{\text{lat}}$\;
  \If{$P_{\mathrm{adj}} < P^*_{\mathrm{adj}}$}{
    $n^* \gets n$;\ \ $P^*_{\text{adj}} \gets P_{\text{adj}}$\;
  }
}
\If{$n^* \neq \textsc{None}$}{
  Forward image $i$ to satellite $n^*$\;
  Credit scientific value to originating satellite\;
}
\end{algorithm}

\end{document}